\documentclass[11pt]{article}  
\usepackage{graphicx}
\usepackage{amsthm,amssymb,amsmath}
\usepackage{wrapfig}
\usepackage[dvipsnames]{xcolor}
\usepackage{xspace}
\usepackage{enumitem}
\usepackage{tabularx}
\usepackage{hyperref}
\usepackage{comment}
\usepackage[margin=1in]{geometry}
\usepackage{authblk}

\newcommand*\patchAmsMathEnvironmentForLineno[1]{%
 \expandafter\let\csname old#1\expandafter\endcsname\csname #1\endcsname
 \expandafter\let\csname oldend#1\expandafter\endcsname\csname end#1\endcsname
 \renewenvironment{#1}%
   {\linenomath\csname old#1\endcsname}%
   {\csname oldend#1\endcsname\endlinenomath}}%
\newcommand*\patchBothAmsMathEnvironmentsForLineno[1]{%
 \patchAmsMathEnvironmentForLineno{#1}%
 \patchAmsMathEnvironmentForLineno{#1*}}%
\AtBeginDocument{%
 \patchBothAmsMathEnvironmentsForLineno{equation}%
 \patchBothAmsMathEnvironmentsForLineno{align}%
 \patchBothAmsMathEnvironmentsForLineno{flalign}%
 \patchBothAmsMathEnvironmentsForLineno{alignat}%
 \patchBothAmsMathEnvironmentsForLineno{gather}%
 \patchBothAmsMathEnvironmentsForLineno{multline}}

\newtheorem{tool}{Tool}
\newtheorem{observation}{Observation}
\newtheorem{lemma}{Lemma}
\newtheorem{problem}{Problem}
\newtheorem{theorem}{Theorem}
\newtheorem{definition}{Definition}
\newtheorem{corollary}{Corollary}

\newcommand{\thmheadfont}{\bfseries}
\newenvironment{repeatenv}[2]%
  {\smallskip\noindent {\thmheadfont #1~\ref{#2}.}\ \itshape}
  {\itshape}

\renewcommand{\H}{\ensuremath{\mathcal{H}}\xspace}
\newcommand{\BB}{\ensuremath{\mathcal{BB}\xspace}}
\newcommand{\T}{\ensuremath{\mathcal{T}}\xspace}

\begin{document}

\title{Covering a set of line segments with a few squares}

\author[1]{Joachim Gudmundsson}
\author[2]{Mees van de Kerkhof}
\author[3]{Andr{\'e} van Renssen}
\author[4]{Frank Staals}
\author[5]{Lionov Wiratma}
\author[6]{Sampson Wong}

\renewcommand\Affilfont{\fontsize{10}{10.8}\selectfont}

\affil[1]{University of Sydney, Australia, joachim.gudmundsson@sydney.edu.au}
\affil[2]{Utrecht University, Netherlands, m.a.vandekerkhof@uu.nl}
\affil[3]{University of Sydney, Australia, andre.vanrenssen@sydney.edu.au}
\affil[4]{Utrecht University, Netherlands, f.staals@uu.nl}
\affil[5]{Parahyangan Catholic University, Indonesia, lionov@unpar.ac.id}
\affil[6]{University of Sydney, Australia, swon7907@uni.sydney.edu.au}

\date{}

\maketitle

\begin{abstract}
  We study three covering problems in the plane. Our original
  motivation for these problems come from trajectory analysis. The
  first is to decide whether a given set of line segments can be
  covered by up to $k=4$ unit-sized, axis-parallel squares.
  We give linear time algorithms for $k\leq 3$ and an $O(n \log n)$ time algorithm for $k=4$.

  The second is to build a data structure on a trajectory to
  efficiently answer whether any query subtrajectory is coverable by
  up to three unit-sized axis-parallel squares.  For $k=2$ and $k=3$
  we construct data structures of size $O(n\alpha(n)\log n)$ in
  $O(n\alpha(n) \log n)$ time, so that we can test if an arbitrary
  subtrajectory can be $k$-covered in $O(\log n)$ time.

  The third problem is to compute a longest subtrajectory of a given
  trajectory that can be covered by up to two unit-sized axis-parallel
  squares.  We give $O(n2^{\alpha(n)}\log^2 n)$ time algorithms for $k\leq 2$.
\end{abstract}

\section{Introduction}

Geometric covering problems are a classic area of research in
computational geometry. The traditional \emph{geometric set cover
  problem} is to decide whether one can place $k$ axis-parallel
unit-sized squares (or disks) to cover $n$ given points in the
plane. If $k$ is part of the input, the problem is known to be
NP-hard~\cite{DBLP:journals/ipl/FowlerPT81,DBLP:journals/siamcomp/MegiddoS84}. Thus,
efficient algorithms are known only for small values of~$k$. For $k=2$
or~$3$, there are linear time
algorithms~\cite{drezner1987rectangular,DBLP:conf/compgeom/SharirW96},
and for $k=4$ or~$5$, there are $O(n \log n)$ time
algorithms~\cite{DBLP:conf/issac/Nussbaum97,DBLP:journals/ijcga/Segal99}. For
general $k$, the $O(n^{\sqrt k})$ time algorithm for unit-sized
disks~\cite{DBLP:journals/algorithmica/HwangLC93} 
can be simplified and extended
to unit-sized axis-parallel
squares~\cite{DBLP:journals/algorithmica/AgarwalP02}.

Motivated by trajectory analysis, we study a line segment variant of
the geometric set cover problem where the input is a set of~$n$ line
segments. Given a set of line segments, we say it is
\emph{$k$-coverable} if there exist $k$ unit-sized axis-parallel squares
in the plane so that every line segment is in the union of the $k$
squares (we may write coverable to mean $k$-coverable when $k$ is
clear from the context). The first problem we study in this paper is:

\begin{problem}
  \label{prob:decision}
  Decide if a set of line segments is $k$-coverable, for $k \in O(1)$.
\end{problem}

A key difference in the line segment variant and the point variant is that each segment need not be covered by a single square, as long as each segment is covered by the union of the $k$ squares. See Figure~\ref{fig:01_cover_3}.

\begin{figure}[bt]
    \centering
    \begin{minipage}{0.5\textwidth}
        \centering
        \includegraphics{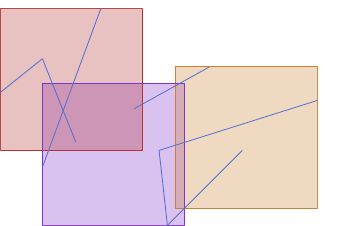}
        \caption{A set of 3-coverable segments.}
        \label{fig:01_cover_3}
    \end{minipage}\hfill
    \begin{minipage}{0.49\textwidth}
        \centering
        \includegraphics{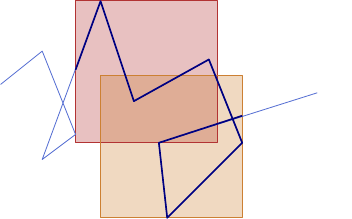}
        \caption{A 2-coverable subtrajectory.}
        \label{fig:03_example}
    \end{minipage}
\end{figure}

Hoffmann~\cite{hoffmann2001covering} provides a linear time algorithm for $k=2$ and $3$, however, a proof was not included in his extended abstract. Sadhu et al.~\cite{DBLP:journals/tcs/SadhuRNR19} provide a linear time algorithm for $k=2$ using constant space. In Section~\ref{sec:problem1}, we provide a proof for a $k=3$ algorithm and a new $O(n \log n)$ time algorithm for~$k=4$. 

Next, we study trajectory coverings. A trajectory $\mathcal{T}$ is a polygonal
curve in the plane parametrised by time. 
Let $v_1,..,v_n$ and $e_1,..,e_{n-1}$ be the vertices and edges of $\T$,
respectively. For two points $s$ and $t$ on \T, we write $s \prec t$
if $s$ occurs on \T before $t$. Any such pair of points then defines a
unique subtrajectory $\T[s,t]$ of \T starting in $s$ and ending in
$t$. See Figure~\ref{fig:03_example} for an example. Trajectories are commonly used to model
the movement of an object (e.g. a bird, a vehicle, etc) through time and space. The analysis of trajectories have
applications in animal ecology~\cite{damiani2014extracting},
meteorology~\cite{stohl2002computation}, and sports
analytics~\cite{gudmundsson2017spatio}. 

To the best of our knowledge, this paper is the first to study
$k$-coverable trajectories for~$k \geq 2$. 

A $k$-coverable trajectory may, for example, model a commonly travelled route, and the squares could model a method of displaying the route (i.e. over multiple pages, or multiple screens), or alternatively, the location of several facilities. We build a data structure that can efficiently decide whether a subtrajectory is $k$-coverable. 

\begin{problem}
    \label{prob:datastructure_decision}
    Construct a data structure on a trajectory, so that given any
    query subtrajectory, it can efficiently answer whether the
    subtrajectory is $k$-coverable, for $k \in O(1)$.
\end{problem}

For $k=2$ and $k=3$ we preprocess a trajectory $\mathcal{T}$ with $n$
vertices in $O(n\log n)$ time, and store it in a data structure of
size $O(n \log n)$, so that we can test if an arbitrary subtrajectory
(not necessarily restricted to vertices) $\mathcal{T}[s,t]$ can be
$k$-covered.

Finally, we consider a natural extension of Problem 2, that is, to calculate the \emph{longest} $k$-coverable subtrajectory of any given trajectory. This problem is similar in spirit to the problem of covering the maximum number of points by $k$ unit-sized axis-parallel squares~\cite{bbdkrs-oscsp-18,mgd-mcti-2008}.

\begin{problem}
  \label{prob:longest}
  Given a trajectory, compute a longest $k$-coverable subtrajectory,
  for $k \in O(1)$.
\end{problem}

Problem~\ref{prob:longest} is closely related to computing a trajectory \emph{hotspot}, which is a small region where a moving object spends a large amount of time. For $k=1$ squares, the existing algorithm by Gudmundsson et~al.~\cite{hotspots2013} computes longest 1-coverable subtrajectory of any given trajectory. We notice a missing case in their algorithm, and show how to resolve this issue in the same running time of~$O(n \log n)$. Finally, we show how to compute the longest 2-coverable subtrajectory of any given trajectory in $O(n 2^{\alpha(n)}\log^2 n)$ time, where $\alpha(n)$ is the extremely slow growing inverse Ackermann function.

\paragraph{Overview} In the next section we consider Problem~\ref{prob:decision}. We build up to the $k=4$ case by first considering the problem for $k\leq 3$. A simple but crucial observation for $k=4$ is that if a set $S$ of segments is $4$-coverable then either (a) one square has to lie in a corner of the bounding box of $S$ or (b) each square has to touch exactly one side of the bounding box of $S$. The first case immediately reduces to the case when $k=3$ which can be solved in linear time, so the focus of Section~\ref{subsec:problem1_k=4} is to solve the second case in $O(n \log n)$ time.

In Section~\ref{sec:problem2} we turn our attention to Problem~\ref{prob:datastructure_decision}. We build four basic data structures (Tools~\ref{tool:bb}--\ref{tool:pl_quadrant}) in $O(n\alpha(n) \log n)$ time that are then combined in Sections~\ref{subsec:problem2_k=2} and~\ref{subsec:problem2_k=3} to produce data structures for $k=2$ and $k=3$, respectively.

Our main technical contributions are in Sections~\ref{sec:A_Longest_1-coverable_subtrajectory} and~\ref{sec:problem3_k=2} where we consider Problem~\ref{prob:longest} for $k\leq 2$. We first note that an earlier algorithm for $k=1$ by Gudmundsson et al.~\cite{hotspots2013} omits a possible scenario. We show how this case can be handled in $O(n \log n)$ time before we show our main result for $k=2$, which is an $O(n 2^{\alpha(n)}\log^2 n)$ time algorithm. 
 
\section{Problem 1: The Decision Problem}
\label{sec:problem1}

We first consider the simple case when $k=2$ to build intuition for the problem and state several basic properties that will be used in later sections.

\subsection{Is a set of line segments 2-coverable?}
\label{subsec:problem1_k=2}

We begin with an observation that applies to any $k$-covering.

\begin{observation}
  \label{obs:covering_touches_all_sides}
  Every $k$-covering of $S$ must touch all four sides of $\BB(S)$.
\end{observation}

The reasoning behind Observation~\ref{obs:covering_touches_all_sides}
is simple: if the covering does not touch one of the four sides, say
the left side, then the covering could not have covered the leftmost
vertex of the set of segments. An intuitive way for two squares to
satisfy Observation~\ref{obs:covering_touches_all_sides} is to place
the two squares in opposite corners of the bounding box. This intuition is formalised in Lemma~\ref{lem:2-covering_corner}.

\begin{lemma}[Sadhu et al.~\cite{DBLP:journals/tcs/SadhuRNR19}]
\label{lem:2-covering_corner}
A set $S$ of segments is 2-coverable if and only if there is a
covering with squares in opposite corners of $\BB(S)$.
\end{lemma}

It suffices to check the two
configurations where squares are in opposite corners of the bounding
box. For each of these two configurations, we simply check if each
segment is in the union of the two squares, which takes linear time in
total, leading to the following theorem:

\begin{theorem}
\label{thm:2-covering}
One can compute a $2$-covering of a set of $n$ line
  segments, or report that no such covering exists, in $O(n)$ time.
\end{theorem}

\subsection{Is a set of line segments 3-coverable?}
\label{subsec:problem1_k=3}

The following lemma is analogous to Lemma~\ref{lem:2-covering_corner}, but for the $k=3$ case.

\begin{lemma}
\label{lem:3-covering_corner}
A set of segments $S$ is 3-coverable if and only if there is a covering
with a square in a corner of the bounding box $\BB(S)$.
\end{lemma}

\begin{proof}
    By Observation~\ref{obs:covering_touches_all_sides}, any 3-covering of $S$ touches all four sides of the bounding box. By the pigeon-hole principle, one of these square must intersect at least two
  sides of $\BB(S)$.

  Consider if these two sides are adjacent. Without
  loss of generality they are the left and top sides of
  $\BB(S)$. If the top-left corner of \H already coincides with the
  top-left corner of $\BB(S)$ the lemma statement holds. If not, the
  top-left corner of \H lies outside of $\BB(S)$, and thus we can
  shift \H to make the two top-left corners coincide.

  This
  increases the area of $\BB(S)$ covered by \H, and hence the
  3-covering remains valid.

  Consider if these two sides are opposite. Without loss of generality they are the top and bottom sides of $\BB(S)$. Consider the square that intersects the
  left side of $\BB(S)$ and shift it to coincide with the top-left
  corner of $\BB(S)$. Since the height of $\BB(S)$ is at most one, we
  still have a valid 3-covering, and this square intersects
  three sides of $\BB(S)$.
\end{proof}

It suffices to consider four cases, one for each corner of the bounding box. After placing the first square in one of the four corners, we subdivide each segment into at most one subsegment that is covered by the first square, and up to two subsegments that are not yet covered. Finally, we use Theorem~\ref{thm:2-covering} to decide whether the final two squares can cover all remaining subsegments.

Subdividing each segment takes linear time in total. There are at most a linear number of remaining subsegments. Checking if the remaining segments are 2-coverable takes linear time by Theorem~\ref{thm:2-covering}. Hence:

\begin{theorem}
\label{thm:3-covering}
One can compute a $3$-covering of a set of $n$ line
  segments, or report that no such covering exists, in $O(n)$ time.
\end{theorem}

\subsection{Is a set of line segments 4-coverable?}
\label{subsec:problem1_k=4}

By Observation~\ref{obs:covering_touches_all_sides}, the four squares of a 4-covering must touch all four sides of the bounding box. We have two cases. In the first case, we have a 4-covering with a square in a corner of the bounding box. In the second case, we have a 4-covering with each square touching exactly one side of the bounding box.

In the first case we can use the same strategy as in the three squares case by placing the first square in a corner and then (recursively) checking if three additional squares can cover the remaining subsegments. This gives a linear time algorithm for the first case.

\begin{figure}[tb]
    \centering
    \begin{minipage}{0.46\textwidth}
        \centering
        \vspace{5pt}
        \includegraphics{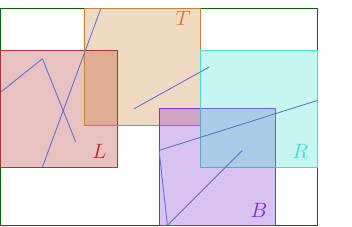}
        \vspace{0pt}
        \caption{The squares $L$, $T$, $B$ and $R$.}
        \label{fig:14_one_per_side}
    \end{minipage}\hfill
    \begin{minipage}{0.53\textwidth}
        \centering
        \includegraphics{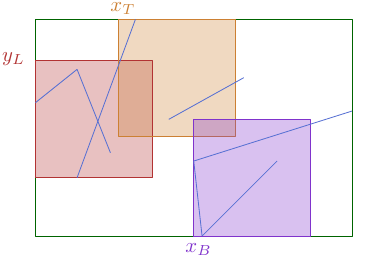}
        \caption{The variables $y_L$, $x_T$ and $x_B$.}
        \label{fig:14_xb_summary}
    \end{minipage}
\end{figure}

For the remainder of this section, we focus on solving the second
case. Define $L$, $B$, $T$, and $R$ to be the square that touches the
left, bottom, top and right side of the bounding box of $S$,
respectively. See Figure~\ref{fig:14_one_per_side}.  Without loss of
generality, suppose that $T$ is to the left of $B$. This implies that
the left to right order of the squares is $L$, $T$, $B$, $R$. Suppose
for now there was a way to compute the initial placement of $L$. Then
we can deduce the position of $T$ as follows.

\begin{lemma}
  \label{lem:placing_T}
  Given the position of $L$, if three additional squares can be placed
  to cover the remaining subsegments, then it can be done with $T$ in
  the top-left corner of the bounding box of the remaining subsegments.
\end{lemma}

  \begin{figure}[tb]
    \centering
    \includegraphics{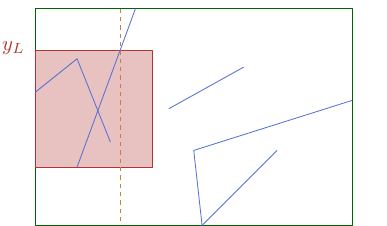}
    \caption{The best position for $T$ is in the top left corner of
      the bounding box of the remaining subsegments.}
    \label{fig:14_xt_10}
  \end{figure}

\begin{proof}
  Consider the bounding box of the subsegments not covered by $L$, in
  particular, the left side of said bounding box. By definition, there
  exists a left endpoint of a remaining subsegment that lies on this
  left side $\ell$ (the dashed line segment in
  Figure~\ref{fig:14_xt_10}). Suppose that the left side of $T$ is to
  the right of $\ell$: then the squares $B$ and $R$ would be even
  further right than $T$, and no square would cover this left endpoint
  of a remaining subsegment. Hence, if a covering exists, the left
  side of $T$ cannot be to the right of $\ell$.

  Suppose that there exists a covering where the left side of $T$ is
  to the left of $\ell$. Then we can shift $T$ so that its left side
  is aligned with $\ell$. As in the proof of
  Lemma~\ref{lem:3-covering_corner} the new position covers more of
  the area of inside the bounding box of the remaining subsegments,
  and hence we still have a valid covering. Furthermore, since $T$ is
  also incident to the top-side of $\BB(S)$ (and $L$ is not incident
  to $\BB(S)$) it follows that the top-left corner of $T$ now
  coincides with the top-side of the bounding box of the remaining
  subsegments as desired.
\end{proof}

After placing the first two squares, we can place
$B$ in the bottom-left corner of the bounding box of the remaining
segments, for reasons analogous to Lemma~\ref{lem:placing_T}. Finally, we cover the remaining segments with $R$, if
possible.

It follows that the position of $L$ along the left boundary uniquely
determines the positions of the squares $T$, $B$ and $R$ along their
respective boundaries. Unfortunately, we do not know the position of
$L$ in advance, so instead we consider all possible initial positions
of $L$ via parametrisation. Let $y_L$ be the $y$-coordinate of the top
side of $L$, and similarly let $x_T$, $x_B$ be the $x$-coordinates of
the left side of $T$ and $B$, respectively. See
Figure~\ref{fig:14_xb_summary}.

Finally, we will try to cover all remaining subsegments with the
square $R$. Define $x_{R_1}$ and $x_{R_2}$ to be the $x$-coordinates
of the leftmost and rightmost uncovered points after the first three
squares have been placed. Similarly, define $y_{R_1}$ and $y_{R_2}$ to
be the $y$-coordinates of the topmost and bottommost uncovered
points. Then it is possible to cover the remaining segments with $R$
if and only if $x_{R_1} - x_{R_2} \leq 1$ and
$y_{R_1} - y_{R_2} \leq 1$.

Since the position of $L$ uniquely determines $T$, $B$ and $R$, we can
deduce that the variables $x_T$, $x_B$, $x_{R_1}$, $x_{R_2}$,
$y_{R_1}$ and $y_{R_2}$ are all functions of $y_L$. We will show that
each of these functions is piecewise linear and can be computed in
$O(n \log n)$ time. We begin by computing $x_T$ as a function of
variable $y_L$.

Let $s \in S$ be a segment, we define
$f_s(y) = \min_{p \in s \land p_y \geq y} p_x$ as the leftmost point
on $s$ above the horizontal line at height $y$, and
$f(y) = \min_{s \in S} f_s(y)$ as the minimum over all segments
$s$. We refer to (the graph of) $f$ as the \emph{skyline} of
$S$.

\begin{lemma}
  \label{lem:skyline}
  The skyline $f$ of a set $S$ of $n$ segments is a piecewise linear, monotonically increasing
  function with $O(n \alpha(n))$ pieces and can be computed in
  $O(n \log n)$ time, where $\alpha(n)$ is the inverse Ackermann
  function. 
\end{lemma}

  \begin{figure}[tb]
    \centering
    \includegraphics[width=0.3\textwidth]{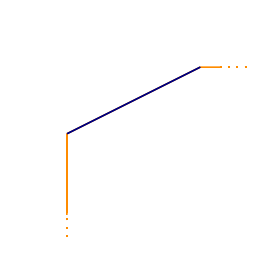}
    \includegraphics[width=0.3\textwidth]{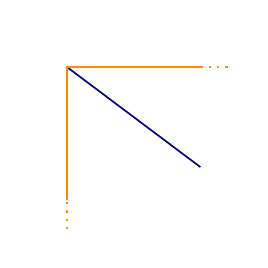}
    \includegraphics[width=0.35\textwidth]{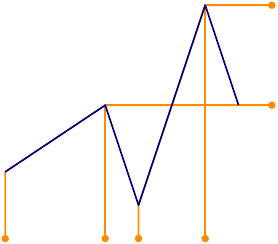}
    \caption{The skyline of a segment with positive gradient (left), a segment with negative gradient (middle) and a set of segments (right).
    }
    \label{fig:skyline}
  \end{figure}

\begin{proof}
  We first compute the skyline of each individual segments. If the segment has positive gradient, then the skyline consists of three pieces. If the segment has negative gradient, then the skyline has two pieces. See Figure~\ref{fig:skyline}, (left) and (middle). Next, we merge the individual skylines into a combined skyline. For any
  $\ell$, the function $f(\ell)$ takes the value of the leftmost
  intersection of the individual skylines with $\ell$, i.e. the upper envelope except in the leftwards cardinal direction
  rather than the upwards direction. See
  Figure~\ref{fig:skyline}, (right). 

    The skylines of individual segments can be computed in $O(n)$
  time, and have a total size of $O(n)$. The leftwards envelope of the individual skylines has at most
  $O(n \alpha(n))$ pieces and can be computed in $O(n \log n)$
  time~\cite{DBLP:books/daglib/davenportschinzel}. The leftwards envelope is piecewise linear and a monotonically increasing function.
\end{proof}

Now we can apply Lemma~\ref{lem:skyline} to compute $x_T$ as a function of $y_L$.

\begin{lemma}
  \label{lem:x_t}
  The variable $x_T$ as a function of variable $y_L$ is a piecewise
  linear, monotonically increasing function of complexity $O(n\alpha(n))$ and can be computed in
  $O(n \log n)$ time.
\end{lemma}

\begin{figure}[tb]
    \centering
        \includegraphics{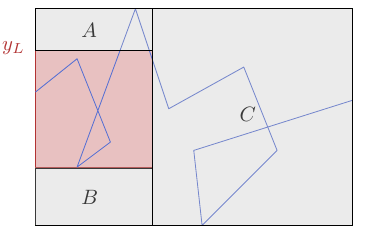}
        \includegraphics[width=0.35\textwidth]{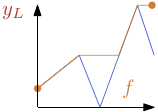}
        \caption{The regions $A$, $B$ and $C$ are above, below, and to the left of the square~$L$ (left). The skyline is the leftmost point of region
          $A$ as $y_L$ increases (right).}
        \label{fig:xt}
\end{figure}

\begin{proof}
We divide the region inside the bounding box and not covered by $L$ into three subregions. The region $A$ is above $L$, the region $B$ is below $L$ and the region $C$ is to the right of $L$. See Figure~\ref{fig:xt} (left). As $y_L$ increases, the leftmost point of $A$ follows the skyline of the segments in $A$. See Figure~\ref{fig:xt}, (right). Analogously, as $y_L$ increases, the leftmost point of $B$ follows the skyline of the segments in $B$, except that the skyline is taken in the downward direction instead. Finally, as $y_L$ increases, the leftmost point of $C$ is a constant. By Lemma~\ref{lem:skyline}, the leftmost points of $A$, $B$ and $C$ are all piecewise linear functions with complexities $O(n \alpha(n))$ that can be computed in $O(n \log n)$ time. The value of $x_T$ is the minimum of these functions, so is a piecewise linear function of complexity $O(n\alpha(n))$ and can be computed in $O(n \log n)$ time.
\end{proof}

Next, we show that $x_B$ is a piecewise linear function of $y_L$, with
complexity $O(n\alpha(n))$, and can be computed in $O(n \log n)$ time.

\begin{lemma}
  \label{lem:x_b}
  The variable $x_B$ as a function of variable $y_L$ is a piecewise
  linear function of complexity $O(n\alpha(n))$ and can be computed in
  $O(n \log n)$ time.
\end{lemma}

\begin{figure}[tb]
  \centering
  \includegraphics{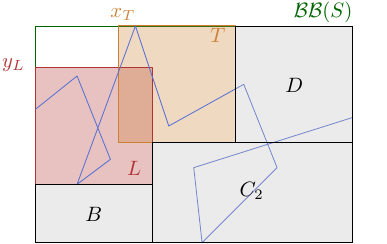}
  \caption{The regions $B$, $D$ and $C_2$ as defined by $L$ and $T$.}
  \label{fig:four_hotspot_b_d_c2}
\end{figure}

\begin{proof}
For an analogous reason to Lemma~\ref{lem:placing_T}, we can place $B$ in the bottom-right corner of the remaining segments after placing $L$ and $T$. Therefore, $x_B$ is the $x$-coordinate of the leftmost uncovered point after placing $L$ and $T$. Divide the remaining region into three subregions: $B$ below $L$, $D$ to the right of $T$, and $C_2$ to the right of $L$ and below $T$. See Figure~\ref{fig:four_hotspot_b_d_c2}. 

As $y_L$ increases, the leftmost point of $B$ moves monotonically to
the left, and follows the skyline of the set of segments, except that
the skyline is taken in the downwards direction. 
By Lemma~\ref{lem:skyline}, the leftmost point of $B$ is a piecewise linear
function in terms of $y_L$ and can be computed in $O(n \log n)$ time.

As $y_L$ increases, the position of square $T$, given by the variable $x_T$ moves monotonically to the right, and follows a skyline of the segments. We have shown in Lemma~\ref{lem:x_t} that $x_T$ is a piecewise linear function in terms of $y_L$ and can be computed in $O(n \log n)$ time. Similarly, the leftmost point of $D$ is a piecewise linear function in terms of the position of $T$ and can be computed in $O(n \log n)$ time. We compose the leftmost point of $D$ as a piecewise linear function of $x_T$, and $x_T$ as a piecewise linear function of $y_L$. By Lemma~\ref{lem:skyline}, each of these functions are monotonically increasing functions with complexity $O(n \alpha(n))$. Composing two monotonic, piecewise linear functions requires a single simultaneous sweep of the two functions, which takes $O(n \alpha(n))$ time. Hence, the leftmost point of $D$ is a piecewise linear function and can be computed in $O(n \log n)$ time.

As $y_L$ increases, the region $C_2$ remains constant, so the leftmost point of $C_2$ remains constant. This point can be computed in $O(n)$ time.

Putting this all together, each of the leftmost points of $B$, $D$ and $C_2$ are piecewise linear functions in terms of $y_L$ and can be computed in $O(n \log n)$ time. Hence, their minimum $x_B$ is a piecewise linear function of complexity $O(n\alpha(n))$ and can be computed in $O(n \log n)$ time.
\end{proof}

Then we compute $x_{R_1}$, $x_{R_2}$, $y_{R_1}$ and $y_{R_2}$ in a similar fashion. 

\begin{lemma}
\label{lem:y_r}
The variables $x_{R_1}$, $x_{R_2}$, $y_{R_1}$, $y_{R_2}$ as functions
of variable $y_L$ are piecewise linear functions of complexity
$O(n\alpha(n))$ and can be computed in $O(n \log n)$ time.
\end{lemma}

\begin{figure}[tb]
  \centering
  \includegraphics{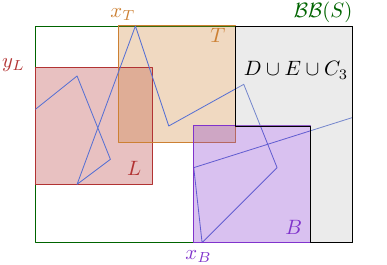}
  \caption{The regions $D \cup E \cup C_3$ as defined by $L$, $T$ and $B$.}
  \label{fig:four_hotspot_d_e_c3}
\end{figure}

\begin{proof}
We will show that $y_{R_1}$ as a function of $y_L$ is a piecewise linear function and can be computed in $O(n \log n)$ time. The other variables follow analogously.

Divide the region to the right of $L$, $T$ and $B$ into three subregions: $D$ is to the right of $L$ and above $B$, $E$ to the right of $B$, and $C_3$ below $D$ and above $B$, if it exists. Note that $C_3$ only exists in squares where the height of the bounding box is greater than two. Otherwise, only the regions $D$ and $E$ exist, as shown in Figure~\ref{fig:four_hotspot_d_e_c3}. We compute the topmost points of $D$, $E$ and $C_3$ separately, then return their overall topmost point to be the value of $y_{R_1}$. 

As $y_L$ increases, the position of square $T$ given by the variable $x_T$ moves monotonically to the right and follows the skyline of the segments. Similarly, as $x_T$ moves monotonically to the right, the topmost point of $D$ moves monotonically to the right. We have shown that each of these monotonic functions are piecewise linear, have complexity $O(n \alpha(n))$ and can be computed in $O(n \log n)$ time (See Lemma~\ref{lem:skyline}). Computing their composition of monotonic, piecewise linear functions requires a single simultaneous sweep of the two functions, which takes $O(n \alpha(n))$. Hence, the overall function is piecewise linear and can be computed in $O(n \log n)$.

As $y_L$ increases, the topmost point of $E$ is similarly a composition of two skyline functions, which is piecewise linear and can be computed in $O(n \log n)$.

Finally, as $y_L$ increases, the topmost point of $C_3$ (if it exists) is constant and can be computed in $O(n)$ time.

Putting this all together, each of the leftmost points of $D$, $E$ and $C_3$ are piecewise linear functions in terms of $y_L$ and can be computed in $O(n \log n)$ time. Hence, their minimum $y_{R_1}$ is a piecewise linear function of complexity $O(n\alpha(n))$ and can be computed in $O(n \log n)$ time.
\end{proof}

Finally, we check if there exists a value of $y_L$ so that $x_{R_1}-x_{R_2} \leq 1$ and $y_{R_1}-y_{R_2} \leq 1$. If so, there exist positions for $L$, $B$, $T$ and $R$ that cover all the segments, otherwise, there is no such position. This yields the following result:

\begin{theorem}
\label{thm:4-covering}
One can decide if a set of $n$ segments is 4-coverable in $O(n \log n)$ time.
\end{theorem}

\section{Problem 2: The Subtrajectory Data Structure Problem}
\label{sec:problem2}

Let $\mathcal T$ be a piecewise linear trajectory of complexity $n$. We briefly describe some tools, then we use our tools to construct data structures for answering if a subtrajectory is either 2-coverable or 3-coverable. 

\begin{tool}
  \label{tool:bb}
  A linear size ``bounding box'' data structure that can be built in
  $O(n)$ time and given a pair of query points $s \prec t$ on \T can
  return the bounding box of $\T[s,t]$ in $O(\log n)$ time.
\end{tool}

\begin{figure}[tb]
    \centering
    \includegraphics{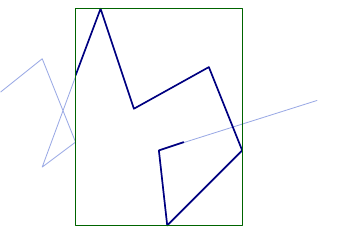}
    \caption{The subtrajectory data structure that returns the bounding box.}
    \label{fig:bb}
\end{figure}

\begin{proof}
Construct a binary tree on the sequence of edges in $\mathcal T$. For each node in the binary tree, the segments that are its descendants form a contiguous subtrajectory of $\mathcal T$. We call such a contiguous subtrajectory a canonical subset. For each internal node, we compute the bounding box of its canonical subset. If we compute the bounding boxes in a bottom up fashion, each internal node can be processed in constant time. The overall construction time is $O(n)$.

Given a query subtrajectory, we decompose the subtrajectory into canonical subsets. We query the data structure to obtain an individual bounding box for each canonical subset. We compute a combined bounding box that contains all individual bounding boxes. The overall query time is $O(\log n)$.
\end{proof}

\begin{tool}
\label{tool:ue}
An $O(n\alpha(n)\log n)$ size ``upper envelope'' data structure that
can be built in $O(n\alpha(n)\log n)$ time and given a pair of query
points $s \prec t$ on \T and a vertical line can return the highest
intersection between the line and subtrajectory $\T[s,t]$ (if one
exists) in $O(\log n)$ time. See Figure~\ref{fig:tool:ue}.
\end{tool}

\begin{proof}

We use an approach similar to Tool~\ref{tool:bb}. We build a binary search tree over segments of $\mathcal{T}$, and associate with each internal node the subset of all its descendants. We call this the canonical subset associated with this internal node. By construction, every subtrajectory can be decomposed into a union of $O(\log n)$ canonical subsets.

At each internal node, we store the upper envelope of all segments in
its canonical subset. This upper envelope can be represented by an
list of its vertices in left-to-right order. Since the upper envelope
of $n$ line segments has size $O(n\alpha(n))$, the total size of our
data structure is $O(n\alpha(n)\log n)$. Computing all these upper
envelopes can be done in $O(n\alpha(n)\log n)$ time using a bottom up
divide and conquer approach~\cite{DBLP:journals/ipl/Hershberger89}.

Given a query subtrajectory and a vertical line at $x$-coordinate $x$,
we can naively answer the upper envelope query in two steps. First, we
decompose the subtrajectory into $O(\log n)$ canonical subsets, and
find the segment realizing the upper envelope at $x$ using a binary
search in $O(\log(n\alpha(n)))=O(\log n)$ time. Second, we report the
highest intersection point of the vertical line with the $O(\log n)$
segments found. The overall query time for this approach would be
$O(\log^2 n)$ time. A standard application of fractional
cascading~\cite{DBLP:journals/algorithmica/ChazelleG86} reduces the
total query time to $O(\log n)$.

\end{proof}

Using Tool~\ref{tool:ue} we can report whether a vertical half-line intersects a subtrajectory. We build Tool~\ref{tool:ue} in all four cardinal directions. Similarly, we will build Tool~\ref{tool:pl_slab} and Tool~\ref{tool:pl_quadrant} in all four cardinal directions.

\begin{figure}[tb]
    \centering
    \begin{minipage}{0.49\textwidth}
        \centering
        \includegraphics{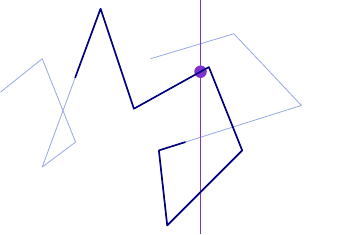}
        \caption{Tool~\ref{tool:ue} returns the highest intersection
          of a subtrajectory and a vertical~line.}
        \label{fig:tool:ue}
    \end{minipage}\hfill
    \begin{minipage}{0.49\textwidth}
        \centering
        \includegraphics{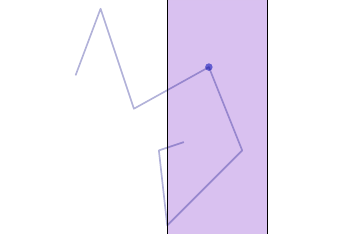}
        \caption{Tool~\ref{tool:pl_slab} returns the highest subtrajectory
          vertex in a query slab.
          }
        \label{fig:tool:pl}
    \end{minipage}
\end{figure}

\begin{tool}
  \label{tool:pl_slab}
  An $O(n\log n)$ size ``highest vertex'' data structure that can be
  built in $O(n\log n)$ time and given a pair of query points
  $s \prec t$ on \T and a vertical query slab can return the highest vertex
  of \T in subtrajectory $\T[s,t]$ in the slab in $O(\log n)$
  time. See Figure~\ref{fig:tool:pl}.
\end{tool}

\begin{proof}
  We store \T in the leaves of a balanced binary search tree in order
  along the trajectory. For each internal node $\nu$, corresponding to
  a canonical subset of vertices $\T_\nu$, we store this set of points
  ordered on increasing $x$-coordinate. In addition, we store their
  $y$-coordinates in an array $A_\nu$ in this order. That is,
  $A_\nu[i]$ stores the $y$-coordinate of the point $p_i$ with the
  $i^\mathrm{th}$-largest $x$-coordinate among $\T_\nu$. Each entry in
  the sequence $\T_\nu$ will store its index in $A_\nu$, moreover we
  assume that for each index $i$ we can again retrieve the point $p_i$
  (e.g. by storing another array that stores the original points). We
  preprocess $A_\nu$ for constant time range maximum
  queries~\cite{bender00lca_probl_revis}, and we build a fractional
  cascading structure on canonical subsets
  $\T_\nu$~\cite{DBLP:journals/algorithmica/ChazelleG86}. Since each
  node uses linear space, the total space required is $O(n\log n)$.
  Moreover, by building the sorted lists $\T_\nu$ in a bottom-up
  fashion we can build the entire data structure in $O(n\log n)$ time
  as well.

  To answer a query, we find the $O(\log n)$ nodes whose canonical
  subsets make up the query subtrajectory. For each such node $\nu$,
  the points from $\T_\nu$ that lie in the query slab are stored
  consecutively in $\T_\nu$ (and $A_\nu$). So, using the fractional
  cascading structure we can find, the index $i_\nu$ of the leftmost
  point among $\T_\nu$ that lies in the vertical query slab. This
  takes $O(\log n)$ time in total. Similarly, we get the index $j_\nu$
  of the rightmost point from $\T_\nu$ in the query slab. We can then
  query the range maximum structure on the array $A_\nu$ with the
  range $i_\nu,..j_\nu$ to find the point with maximum $y$-coordinate
  in constant time. We do this for all $O(\log n)$ nodes and report
  the highest point found. It follows that this is the highest vertex
  on the query subtrajectory that also lies in the query slab.
\end{proof}

\begin{tool}
  \label{tool:pl_quadrant}
  An $O(n\log n)$ size ``highest vertex'' data structure that can be
  built in $O(n\log n)$ time and given a pair of query points
  $s \prec t$ on \T and a query quadrant
  $Q=[Q_x,\infty) \times [Q_y,\infty)$ can return the highest vertex
  of \T in $\T[s,t] \cap Q$ in $O(\log n)$ time.
\end{tool}
\begin{proof}
  We again store the vertices of \T in the leaves of a balanced binary
  search tree in order along the trajectory. For each internal node
  $\nu$, consider the function $f_\nu(x)$ expressing the maximum
  $y$-coordinate among the points in $T_\nu$ right of the vertical
  line at $x$. Observe that this function is piecewise constant,
  monotonically decreasing, and has complexity $O(n)$. We store the
  graph of this function $f_\nu$ by storing an ordered sequence of its
  breakpoints, which allows us to evaluate $f_\nu(x)$ for some value
  $x$ in $O(\log n)$ time by a binary search. The data structure uses
  $O(n \log n)$ space, and can be built in $O(n\log n)$ time
  (e.g. again by sorting the points on increasing $x$-coordinate in a
  bottom up fashion).

  To answer a query, we find the $O(\log n)$ nodes whose canonical
  subsets make up the query subtrajectory. For each such node $\nu$,
  we evaluate $f_\nu(Q_x)$, and compute the maximum over all nodes. If
  this value is at least $Q_y$ the corresponding point is the highest
  vertex of the query subtrajectory in quadrant $Q$, and hence we can
  report it. If the value is smaller than $Q_y$ the quadrant is
  empty. The query time is $O(\log^2 n)$ time, which we can reduce to
  $O(\log n)$ using fractional
  cascading~\cite{DBLP:journals/algorithmica/ChazelleG86}.
\end{proof}

Using a data structure analogous to Tool~\ref{tool:pl_quadrant} we can
also report the lowest point on the query subtrajectory in a quadrant
$[Q_x,\infty) \times [Q_y,\infty)$.

\subsection{Query if a subtrajectory is 2-coverable}
\label{subsec:problem2_k=2}

We start with a lemma to help us apply Tool~\ref{tool:ue} to the boundary of the union of two axis aligned unit squares.

\newcommand{\U}{\ensuremath{\mathcal{U}}\xspace}
\begin{lemma}
  \label{lem:union_2squares}
  The union $\U = \H_1 \cup \H_2$ be of two axis aligned unit squares
  has constant complexity, every edge $\overline{uv}$ of \U is axis
  aligned, and at least one of the half-lines $\overrightarrow{uv}$
  and $\overrightarrow{vu}$ does not intersect the interior of \U. See Figure~\ref{fig:lemma9_figure}, (left).
\end{lemma}

\begin{figure}[tb]
    \centering
    \includegraphics{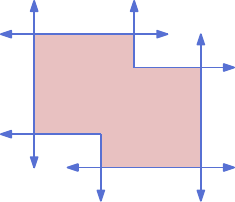}
    \includegraphics{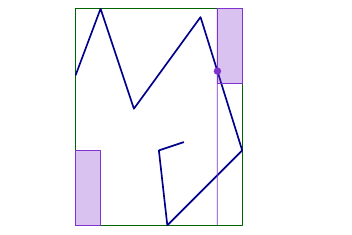}
    \caption{The union of two squares and the half-lines that do not intersect the interior (left). One of the four half-lines on the boundary of the union of two squares (right).}
    \label{fig:lemma9_figure}
\end{figure}

\begin{proof}
  The only non-trivial part of the Lemma is to argue that for every
  edge $\overline{uv}$ either $\overrightarrow{uv}$ or
  $\overrightarrow{vu}$ does not intersect the interior of \U. Assume
  w.l.o.g. that $\overline{uv}$ is part of the boundary of $\H_1$ (and
  thus lies outside of $\H_2$), and assume that $\overrightarrow{uv}$
  intersects the interior of $\U$ (otherwise the claim already
  holds). It follows that $\overrightarrow{uv}$ intersects the
  interior of $\H_2$. By convexity of $\H_2$ it then follows
  $\overrightarrow{vu}$ does not intersect the interior of $\H_2$,
  otherwise $\overline{uv}$ would be contained in $\H_2$. This
  completes the proof.
\end{proof}

Now we are ready to prove the main result of Section~\ref{subsec:problem2_k=2}.

\begin{theorem}
  \label{thm:problem2_k=2}
  Let $\mathcal{T}$ be a trajectory with $n$ vertices. After $O(n\alpha(n)\log n)$
  preprocessing time, $\mathcal{T}$ can be stored using $O(n\alpha(n)\log n)$ space, so
  that deciding if a query subtrajectory $\mathcal{T}[a,b]$ is 2-coverable
  takes $O(\log n)$ time.
\end{theorem}

\begin{proof}
  Our construction procedure is to build Tool~\ref{tool:bb}, and Tool~\ref{tool:ue} for all four cardinal directions. Our query procedure consists of three steps. First, we use Tool~\ref{tool:bb} to compute the bounding box of the subtrajectory. Second, we apply Lemma~\ref{lem:2-covering_corner} to obtain a configuration of two squares. Finally, we check if the subtrajectory is inside the union of the two squares. We do so by using Tool~\ref{tool:ue} to check if the subtrajectory passes through any of the four half-lines on the boundary of the union, see Figure~\ref{fig:lemma9_figure}, (right). The construction procedure takes $O(n\alpha(n)\log n)$ time and space. The query procedure takes $O(\log n)$ time.
\end{proof}

\subsection{Query if a subtrajectory is 3-coverable}
\label{subsec:problem2_k=3}

We start with a lemma analogous to Lemma~\ref{lem:union_2squares}, but for three squares.

\begin{lemma}
  \label{lem:union_3squares}
  The union $\U = \H_1 \cup \H_2 \cup \H_3$ be of three axis aligned
  unit squares has constant complexity, every edge $\overline{uv}$ of
  \U is axis aligned, and there is at most one edge $\overline{uv}$
  for which both the half-lines $\overrightarrow{uv}$ and
  $\overrightarrow{vu}$ both intersect the interior of \U. See Figure~\ref{fig:extend_indefinite}.
\end{lemma}

\begin{figure}[tb]
  \centering
  \includegraphics{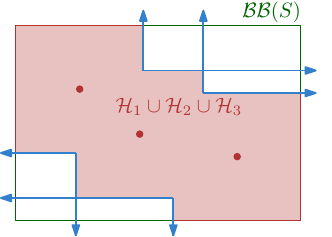}
  \includegraphics{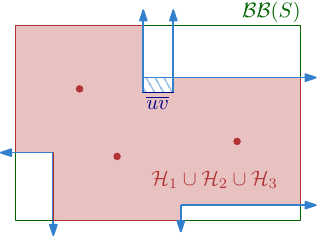}
  \caption{The boundary pieces of the union and the cardinal
    directions we can extend the boundary piece
    indefinitely.
  }
  \label{fig:extend_indefinite}
\end{figure}

\begin{proof}
  Assume, by contradiction, that there are two edges
  $\overline{u_1v_1}$ and $\overline{u_2v_2}$ for which both rays
  intersect the interior of \U.

  Assume without loss of generality that $\overline{u_2v_2}$ lies on
  the top side of $\H_2$ (otherwise rotate the plane), and that $u_2$
  lies left of $v_2$. Since the squares are all convex, it now follows
  that: (i) $\overrightarrow{u_2v_2}$ and $\overrightarrow{v_2u_2}$ do
  not intersect the interior of $\H_2$, and (ii) that they cannot both
  intersect the same $\H_i$. So, assume without loss of generality
  that $\overrightarrow{u_2v_2}$ intersects $\H_3$ and that
  $\overrightarrow{v_2u_2}$ intersects $\H_1$. It then follows that
  the left to right ordering of (the centers of) the squares is
  $\H_1,\H_2,\H_3$. Furthermore, (the center of) $\H_2$ is the lowest
  center among the three squares.

  We now argue that $\overline{u_1v_1}$ cannot also be horizontal. Via
  the same reasoning as above, $\overrightarrow{u_1v_1}$ and
  $\overrightarrow{v_1u_1}$ must hit different squares, and thus
  $\overline{u_1v_1}$ must lie on the middle square $\H_2$. In
  particular, since $\H_2$ is the lowest square and the squares have
  the same size, $\overline{v_1u_1}$, it must lie on the top side of
  $\H_2$ as well (the horizontal line through the bottom side does not
  intersect $\H_1$ or $\H_3$). That implies that two oppositely
  oriented rays (e.g. $\overrightarrow{u_2v_2}$ and
  $\overrightarrow{v_1u_1}$) intersect the same square
  (e.g. $\H_3$). However, since all squares have the same size we
  again get a contradiction.

  So $\overline{u_1v_1}$ is vertical, with say $u_1$ below $v_1$. We
  once again have that $\overrightarrow{u_1v_1}$ and
  $\overrightarrow{v_1u_1}$ must intersect different
  squares. Moreover, the centers of these squares must lie on the same
  of the vertical line through $\overline{u_1v_1}$. Consider the case
  that these centers lie right of this line. The other case is
  symmetric. It now follows that $\overline{u_1v_1}$ must lie on the
  leftmost square $\H_1$, and that the square below
  $\overline{u_1v_1}$ must be $\H_2$. More specifically, $u_1$ must
  lie on or above the top side of $\H_2$ (again since the squares have
  the same size). Similarly, $v_1$ must lie below the bottom side of
  $\H_3$. However, this implies that the bottom side of $\H_3$ lies
  above the top side of $\H_2$, and thus the horizontal ray
  $\overrightarrow{u_2v_2}$ does not intersect the interior of
  $\H_3$. Contradiction.
\end{proof}

Now we can prove the main result of Section~\ref{subsec:problem2_k=3}.

\begin{theorem}
  \label{thm:problem2_k=3}
  Let $\mathcal{T}$ be a trajectory with $n$ vertices. After
  $O(n\alpha(n)\log n)$ preprocessing time, $\mathcal{T}$ can be stored using
  $O(n\alpha(n)\log n)$ space, so that deciding if a query subtrajectory
  $\T[a,b]$ is 3-coverable takes $O(\log n)$ time.
\end{theorem}

  \begin{figure}[tb]
    \centering
    \includegraphics{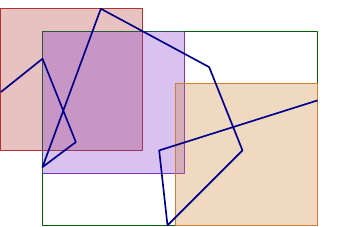}
    \includegraphics{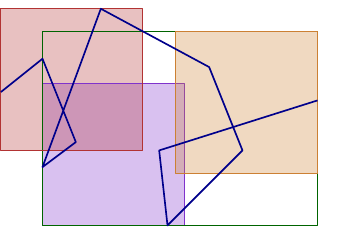}
    \caption{The second and third squares in one of two configurations.}
    \label{fig:23_pos00}
  \end{figure}

\begin{proof}
  Our construction procedure is to build
  Tool~\ref{tool:bb}, and Tools~\ref{tool:ue},~\ref{tool:pl_slab},
  and~\ref{tool:pl_quadrant} in all four cardinal directions. The total  preprocessing time and space usage is $O(n\alpha(n)\log n)$.

  Our query procedure consists of three
  steps. The first step is to place the first square in the corner of the bounding box. The second step is to place the
  second and third squares in one of two configurations, see Figure~\ref{fig:23_pos00}. The
  third step is to check if the configuration of three squares covers
  the subtrajectory.

  In the first step, we use Tool~\ref{tool:bb} to compute the bounding
  box of the query subtrajectory, and then apply
  Lemma~\ref{lem:3-covering_corner} to place the first square in
  a corner of the bounding box.

  In the second step, we compute the bounding box of the uncovered
  segments after placing the first box. Then we apply
  Lemma~\ref{lem:2-covering_corner} to place the final two squares in
  opposite corners of the bounding box of the uncovered
  subsegments. See Figure~\ref{fig:23_pos00}. Suppose we placed a square in the top-left corner in the first
  step. We have two cases for the topmost uncovered point: it is either
  a vertex of the subtrajectory, or the intersection of a subtrajectory
  edge with the left or bottom side of the top-left square. The first
  case can be handled by two queries to Tool~\ref{tool:pl_quadrant}. The second case can be handled by querying 
  Tool~\ref{tool:ue} along the right or bottom boundaries of the
  top-left square, and taking the highest of these points. We apply the same procedure in all four cardinal directions to obtain the bounding box of the uncovered subsegments, as required. 

  For the third step, we check if a given configuration
  $\U=\H_1 \cup \H_2 \cup \H_3$ of three squares covers the subtrajectory
  $\T[a,b]$. The approach is similar to the two square case: we check if
  the starting point $a$ lies inside \U, and if the subtrajectory ever
  exits \U. We use a combination of Tools~\ref{tool:ue},~\ref{tool:pl_slab} and~\ref{tool:pl_quadrant}, to~achieve~this.
  
  By Lemma~\ref{lem:union_3squares} there is at most one edge
  $\overline{uv}$ on the boundary of \U for which both rays
  $\overrightarrow{uv}$ and $\overrightarrow{vu}$ hit the (interior of)
  \U. Hence, for all edges other than $\overline{uv}$ we can check if the
  subtrajectory exits \U using an appropriate copy of
  Tool~\ref{tool:ue}. Observe that if the
  subtrajectory does not intersect the boundary of \U in any edge other than $\overline{uv}$
  then either it exits \U through $\overline{uv}$ and does not return,
  or it exits and reenters \U through $\overline{uv}$. In both cases a
  vertex of $\T[a,b]$ (possibly its endpoint) must lie in the connected
  component of $\BB(\T[a,b])\setminus \U$ that is incident to
  $\overline{uv}$. We can check this using Tools~\ref{tool:pl_slab} and~\ref{tool:pl_quadrant}. 
  
  All in all we make a constant number of queries to
  Tools~\ref{tool:bb},~\ref{tool:ue},~\ref{tool:pl_slab},
  and~\ref{tool:pl_quadrant}, each of which takes $O(\log n)$
  time. Hence, we can test if a query subtrajectory is 3-coverable in
  $O(\log n)$ time.
\end{proof}

\section{Problem 3 for \texorpdfstring{$k=1$}{k=1}: A Longest 1-Coverable Subtrajectory}
\label{sec:A_Longest_1-coverable_subtrajectory}

In this section we compute a longest $k$-coverable subtrajectory
$\T[p^*,q^*]$ of a given trajectory $\T$ for $k=1$. Note that the start and end
points $p^*$ and $q^*$ of such a subtrajectory need not be vertices of
the original trajectory. Gudmundsson, van Kreveld, and
Staals~\cite{hotspots2013} presented an $O(n \log n)$ time algorithm
for the case $k=1$. However, we note that there is a mistake in
one of their proofs, and hence their algorithm misses one of the
possible scenarios. We show how this case can also be handled in
$O(n\log n)$ time, thus correcting their mistake.

\begin{wrapfigure}[13]{r}{0.35\textwidth}
  \centering
  \includegraphics{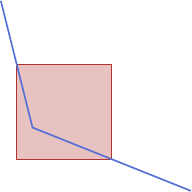}
  \caption{An optimal placement that has no vertex on the boundary of
    the square.}
  \label{fig:counterexample}
\end{wrapfigure}

Gudmundsson, van Kreveld, and Staals state that there exists an
optimal placement of a unit square, i.e. one such that the square
covers a longest $1$-coverable subtrajectory of $\T$, and has a vertex
of $\T$ on its boundary~\cite[Lemma 7]{hotspots2013}. However, that is
incorrect, as illustrated in Figure~\ref{fig:counterexample}. Let $p(t)$ be a parametrisation of the trajectory. Fix a corner $c$ of the square and shift the square so that $c$ follows $p(t)$. Let $q(t)$ be the point so that $\mathcal T[p(t),q(t)]$ is the maximal subtrajectory contained in the square, and let $\phi(t)$ be the length of this subtrajectory. This function $\phi$ is piecewise linear, with inflection points not only when a vertex of $\T$ lies on the boundary of the square, but also when $p(t)$ or $q(t)$ hits a corner of the square. The argument in
\cite{hotspots2013} misses this last case. Instead, the correct
characterization is:

\begin{lemma}
  \label{lem:characterization_k=1}
  Given a trajectory $\T$ with vertices $v_1,..,v_n$, there exists a
  square $\H$ covering a longest 1-coverable subtrajectory so that
  either:
\begin{itemize}[noitemsep]
  \item there is a vertex $v_i$ of $\T$ on the boundary of $\H$, or
  \item there are two trajectory edges passing through opposite
    corners of $\H$.
\end{itemize}
\end{lemma}

\begin{proof}
  Proof by contradiction. Assume that $\H^*$ is an optimal square
  (i.e. covering a longest 1-coverable subtrajectory), and there is no
  optimal square satisfying the conditions in the lemma
  statement. Since $\H^*$ is optimal, the longest contiguous
  subtrajectory $\T^* = \T[p^*,q^*]$ in $\H^*$ must touch two
  opposite sides of $\H^*$. Assume without loss of generality that
  these sides are horizontal.

  Let $\H'=\H^*$ and let $\T'$ be a maximal length sub-trajectory in
  $\H'$, such that initially $\T' = \T^*$. It is easy to see that we
  can shift $\H'$ horizontally---while keeping $\T^*$ inside
  it---until either a vertex of $\T'$ lies on (a vertical side of)
  $\partial \H'$ or $p^*$ lies on a corner of $\H'$. In the former
  case we immediately obtain a contradiction. In the latter case,
  translate $\H'$ while keeping the starting point $p'$ of $\T'$ on
  the same corner of $\H'$ (moving the starting point of $\T'$ earlier
  or later).  Let $\phi(t)$ denote the length of $\T'$ as a function
  of the starting time $t=t_{p'}$ of $\T'$. Function $\phi$ has break
  points when: (i) $p'$ or $q'$ crosses a vertex, (ii) $\H'$ gets a
  vertex of $\T'$ on its boundary, or (iii) when the side of $\H'$
  containing $q$ changes. Since $\phi$ is (piecewise) linear, we can
  either increase or decrease $t$ without decreasing $\phi(t)$ until
  $\phi$ is at a break point. At such a break point $\H'$ has a vertex
  of $\T'$ on its boundary (cases (i) and (ii)) or $q$ lies in a
  corner of $\H'$ (case (iii)). In the former case we arrive at a
  contradiction. In the latter case, observe that $p$ also lies in a
  corner of $\H'$ (by definition of $\phi$). If this corner is
  opposite to that of $q$ we satisfy the second condition of the
  lemma, and thus arrive at a contradiction as well. Otherwise, the
  two corners lie on the same side, say the top side, of $\H'$, and
  thus we can shift $\H'$ upwards while covering $\T'$, until a vertex
  of $\T'$ now lies on the bottom side of $\H'$. Hence, we also arrive
  at a contradiction in this final case. This completes the proof.
\end{proof}

To compute a longest 1-coverable subtrajectory we now have two cases, as described by Lemma~\ref{lem:characterization_k=1}. In the first case, i.e. when there is a vertex $v_i$ on the boundary of $H$, we use the existing algorithm of Gudmundsson et al.~\cite{hotspots2013} to compute the longest 1-coverable subtrajectory. It remains to handle the second case, i.e. when there are two trajectory edges passing through opposite corners of $H$. We begin by showing a useful lemma.

\begin{lemma}
  \label{obs:edge_pair_to_hotspot}
  Given a pair of non-parallel edges $e_i$ and $e_j$ of $\T$, there is
  at most one unit square $H$ such that the top left corner of $H$
  lies on $e_i$, and the bottom right corner of $H$ lies on $e_j$.
\end{lemma}

\begin{proof}
Let $p$ be the top-left corner of $H$ and $q$ be the bottom right corner of $H$. Consider as $p$ moves along the edge $e_i$. Then $q$ also moves along a straight segment, $e_i'$, that is a translated copy of $e_i$. By the conditions of the lemma, $q$ also lies on $e_j$. Therefore, $q$ must be the intersection of $e_i'$ and $e_j$, if one exists, and the observation follows.
\end{proof}

It follows that any pair of edges $e_i,e_j$ of $\T$ generates at most a
constant number of additional candidate placements that we have to
consider. Let $\mathcal{H}_{ij}$ denote this set. Next, we argue that there are
only $O(n)$ relevant pairs of edges that we have to consider.

We define the \emph{reach} of a vertex $v_i$, denoted $r(v_i)$, as the vertex $v_j$ such that $\T[v_i,v_j]$ can be 1-covered, but
$\T[v_i,v_{j+1}]$ cannot. Let $\mathcal{H}_i = \mathcal{H}_{(i-1)j}$ denote the set of
candidate placements corresponding to $v_i$ and
$v_j=r(v_i)$. Analogously, we define the \emph{reverse reach}
$\mathit{rr}(v_j)$ of $v_j$ as the vertex $v_i$ such that
$\T[v_i,v_j]$ can be 1-covered, but $\T[v_{i-1},v_j]$ cannot, and the
set $\mathcal{H}'_j = \mathcal{H}_{(i-1)j}$. Finally, let
$\mathcal{H} = \bigcup_{i=1}^n \mathcal{H}_i \cup \mathcal{H}'_i$ be the set of placements
contributed by all reach and reverse reach pairs. By Lemma 12, $H_i$ consists of at most one element. Similarly, $H_i'$ consists of at most one element. Therefore, $H$ is the union of $2n$ sets, each with at most one element, so $|H| = O(n)$.

\begin{lemma}
  \label{lem:candidate_starting_points_one_hotspot}
  Let $p^* \in e_i$ and $q^* \in e_j$ lie on edges of $\T$, and let $H$
  be a unit square with $p^*$ in one corner, and $q^*$ in the opposite
  corner. We have that $H \in \mathcal{H}$.
\end{lemma}

\begin{proof}
  Observe that vertices $v_{i+1}$ and $v_j$ are inside $H$ whereas
  $v_i$ and $v_{j+1}$ are outside of $H$. See
  Figure~\ref{fig:uvwx}. We now distinguish between two cases,
  depending on whether $v_j$ is reachable from $v_i$ or not.

  \begin{figure}[t]
    \centering
    \includegraphics[width=0.3\textwidth]{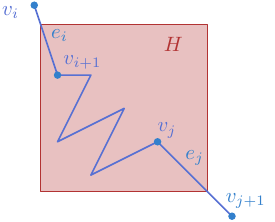}
    \caption{An optimal hotspot $H$.}
    \label{fig:uvwx}
  \end{figure}

  If $v_j$ is reachable from $v_i$ then $v_{j+1}$ is not (otherwise
  $\T[v_i,v_{j+1}] \supset \T[p^*,q^*]$ is 1-coverable, and thus
  $\T[p^*,q^*]$ is not a longest 1-coverable subtrajectory). Hence,
  $v_j=r(v_i)$ is the reach of $v_i$, and thus
  $H \in \mathcal{H}_i \subseteq \mathcal{H}$.

  If $v_j$ is not reachable from $v_i$, then $\T[v_i,v_j]$ cannot be
  1-covered. However, since $v_{i+1}$ and $v_j$ are contained in $H$
  the subtrajectory $\T[v_{i+1},v_j]$ \emph{can} be 1-covered. It
  follows that $v_{i+1}$ is the reverse reach of $v_j$, and thus
  $H \in \mathcal{H}'_j \subseteq \mathcal{H}$.
\end{proof}

Once we have the reach $r(v_i)$ and the reverse reach
$\mathit{rr}(v_i)$ for every vertex $v_i$ we can easily construct $\mathcal{H}$
in linear time (given a pair of edges $e_i,e_j$ we can construct the
unit squares for which one corner lies on $e_i$ and the opposite
corner lies on $e_j$ in constant time). We can use Tool~\ref{tool:bb}
to test each candidate in $O(\log n)$ time. So all that remains is to
compute the reach of every vertex of $\T$; computing the
reverse reach is analogous.

\begin{lemma}
  \label{lem:compute_reach}
  We can compute $r(v_i)$, for each vertex $v_i \in \T$, in
  $O(n\log n)$ time in total.
\end{lemma}

\begin{proof}
We can prove this result using a sliding window approach. For $v_1$ we just naively
  test the subtrajectories $\T[v_1,v_j]$, starting with $j=1$ until we
  find a $\T[v_1,v_{j+1}]$ that we can no longer cover. Hence
  $r(v_1)=v_j$. To compute the reach of $v_{i+1}$, we now simply
  continue this procedure starting with $v_j=r(v_i)$. In total this
  requires $O(n)$ calls to Tool~\ref{tool:bb}, which take $O(\log n)$
  time each. This proves the result.
\end{proof}

Lemma~\ref{lem:compute_reach} gives the following result.

\begin{theorem}
  \label{thm:longest_1-coverable}
  Given a trajectory $\T$ with $n$ vertices, there is an $O(n \log n)$
  time algorithm to compute a longest 1-coverable subtrajectory of $\T$.
\end{theorem}

\section{Problem 3 for 
\texorpdfstring{$k=2$}{k=2}: A Longest 2-Coverable Subtrajectory}
\label{sec:problem3_k=2}

In this section we reuse some of the observations from
Section~\ref{sec:A_Longest_1-coverable_subtrajectory} to develop an
$O(n2^{\alpha(n)}\log^2 n)$ time algorithm to compute a longest
$k$-coverable subtrajectory for $k=2$. In particular, we will compute
the first such longest $2$-coverable subtrajectory $\T[p^*,q^*]$ of
\T, and the squares $\H_1$ and $\H_2$ that cover $\T[p^*,q^*]$ (and
such that $p^* \in \H_1$). We refer to $\T[p^*,q^*]$ as the optimal
subtrajectory.

Our algorithm to compute $\T[p^*,q^*]$ consists of five steps. In Section~\ref{sub:Characterizing_the_set_of_starting_points}, we construct a discrete set $S$ of candidate starting points
on $\T$. In Section~\ref{subsec:proving_pstar_in_S}, we prove $p^* \in S$, where $p^*$ is the starting point of the optimal trajectory and $S$ is the set of candidate starting points. In Section~\ref{sub:reach_of_a_point}, we generalise the notion of the reach, and we generalise Lemma~\ref{lem:compute_reach} to obtain an algorithm for computing the reach. In Section~\ref{sub:set_of_starting_points} we show how to compute all six types of candidate starting points efficiently. Finally, in Section~\ref{sub:Computing_a_longest_subtrajectory}, we compute the reach of all candidate starting points to obtain the optimal subtrajectory.

\subsection{Identifying the set of starting points}
\label{sub:Characterizing_the_set_of_starting_points}

In this section we identify a discrete set $S$ of candidate starting
points on \T. In the subsequent section we prove $p^* \in S$. We define six types of events, depending on different types of starting points, as follows. Given a trajectory $\T$, $p$ is a

\begin{figure}[tb]
    \centering
    \includegraphics[width=0.3\textwidth]{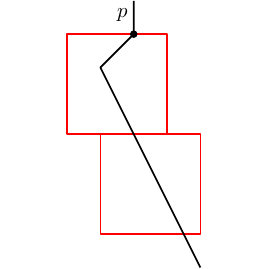}
    \includegraphics[width=0.3\textwidth]{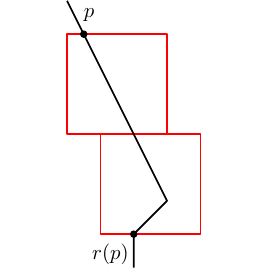}
    \includegraphics[width=0.3\textwidth]{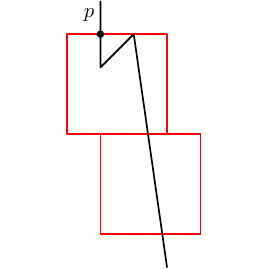}
    \caption{A vertex event (left), a reach event (middle), and a bounding box event (right).}
    \label{fig:vertex_and_reach_event}
\end{figure}

\begin{description}[font=\normalfont]
\item[\textbf{vertex} event]
(see Figure~\ref{fig:vertex_and_reach_event} (left))  if and only if\\
  $p$ is a vertex of \T.
\item[\textbf{reach} event]
(see Figure~\ref{fig:vertex_and_reach_event} (middle))
  if and only if\\
  $r(p)$ is a vertex of $\T$, and no point $q \prec p$ satisfies
  $r(q) = r(p)$.
\item[\textbf{bounding box}] 
event (see Figure~\ref{fig:vertex_and_reach_event}
  (right)) if and only if\\
the topmost vertex of $\T$ within the subtrajectory
  $\T[p, r(p)]$ has the same $y$-coordinate as $p$.
\item[\textbf{bridge} event]
(see
  Figure~\ref{fig:bridge_event_and_upper_envelope_event} (left)) if and
  only if
  \begin{itemize}[nosep]
    \item the point $p$ is the leftmost point on $\T[p, r(p)]$, and
    \item the point $p$ is one unit to the left of a point $u \in \T[p, r(p)]$, and
    \item the point $u$ is one unit above the lowest vertex of $\T$ in the subtrajectory $\T[p, r(p)]$.
    \end{itemize}
\item[\textbf{upper envelope} event]
(see
  Figure~\ref{fig:bridge_event_and_upper_envelope_event}) if and only
  if
  \begin{itemize}[nosep]
    \item the point $p$ is the leftmost point on $\T[p, r(p)]$, and
    \item the point $p$ is one unit to the left of a point $u \in \T[p, r(p)]$, and
    \item the point $u$ is an vertex on the upper envelope of $\T[p, r(p)]$.
    \end{itemize}
\item[\textbf{special configuration} event]
(see
  Figure~\ref{fig:special_configuration_event}) if and only if\\
there is a covering of squares $\H_1$ and $\H_2$ so that $\H_1$ contains the top-left corner of $\H_2$, and either:
\begin{itemize}[nosep]
    \item point $p$ is in the top-right corner of $\H_1$ and $r(p)$ is in the bottom-left corner of $\H_1$, or
    \item point $p$ is in the top-left corner of $\H_1$ and the trajectory $\T$ passes through the bottom-right corner of $\H_1$, or
    \item point $p$ is in the top-left corner of $\H_1$, $r(p)$ is in the bottom-right corner of $\H_2$, and the trajectory $\T$ passes through the two intersections of $\H_1$ and $\H_2$.
    \end{itemize}
\end{description}

\begin{figure}[tb]
    \centering
    \begin{minipage}{0.45\textwidth}
        \centering
        \includegraphics[width=0.7\textwidth]{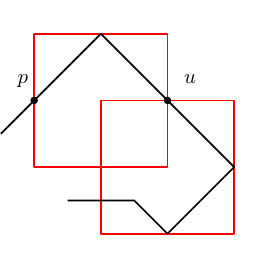}
    \end{minipage}
    \begin{minipage}{0.45\textwidth}
        \centering
    \includegraphics[width=0.7\textwidth]{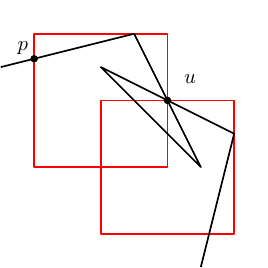}
    \end{minipage}
    \caption{Examples of a bridge event (left), and an upper envelope event box (right).}
    \label{fig:bridge_event_and_upper_envelope_event}
\end{figure}

\begin{figure}[tb]
    \centering
    \includegraphics[width=\textwidth]{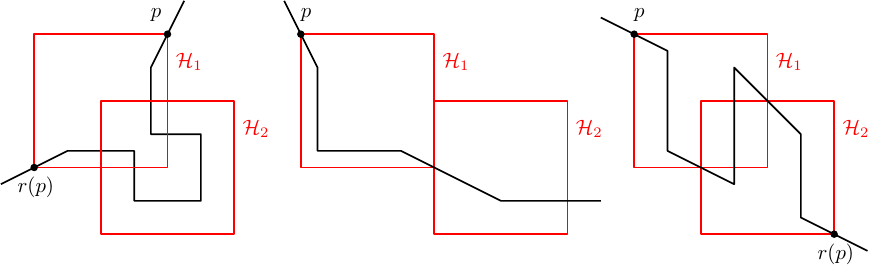}
    \caption{Examples of the three types of special configuration events.}
    \label{fig:special_configuration_event}
\end{figure}

Next, we use these six event types to define the set of candidate starting points in Definition~\ref{defn:starting_points_are_vertices_of_t3}. Note that in Definition~\ref{defn:starting_points_are_vertices_of_t3}, we generalise bounding box, bridge, upper envelope and special configuration events to include the events for all four cardinal directions, not just the upwards cardinal direction. For example, a point $p$ for which
the bottom-most vertex of $\T[p,r(p)]$ also has $y$-coordinate $p_y$
is also a bounding box event.

\begin{definition}
\label{defn:starting_points_are_vertices_of_t3}
Let $\T_1$ be a copy of $\T$ with the following additional points added to the set of vertices of $\T_1$:
\begin{itemize}
    \item all the vertex, reach, bounding box, and bridge events of $\T$ for all four cardinal directions.
\end{itemize}
Next, let $\T_2$ be a copy of $\T_1$ with the following additional points added to the set of  vertices of $\T_2$:
\begin{itemize}
    \item all the upper envelope events of $\T_1$ for all four cardinal directions.
\end{itemize}
Finally, let $\T_3$ be a copy of $\T_2$ with the following additional points added to the set of vertices of $\T_3$:
\begin{itemize}
    \item all the special configuration events of $\T_2$ for all configurations of $\mathcal{H}_1$~and~$\mathcal{H}_2$.
\end{itemize}
\end{definition}

Finally, we define $S$ to be the vertices of $\T_3$. This completes the characterization of $S$, the set of candidate starting points.

\subsection{Proof that \texorpdfstring{$p^* \in S$}{p* in S}}
\label{subsec:proving_pstar_in_S}

In this section, we prove that the starting point $p^*$ of the optimal
subtrajectory is a candidate starting point in $S$. Our main strategy
is to argue that if $p^*$ is not a vertex of $\T_2$, either $p^*$ or
$q^*$ must be in a corner of one of the two squares
(Lemma~\ref{lem:not_corner}). Using a careful analysis, we then argue
that the solution must actually be a special configuration event, and
thus $p^* \in S$. Next, we define bridging points, which help us to
establish some useful technical lemmas.

Given a point $p$, a point $t \in \T[p,r(p)]$ is said to be a
\emph{bridging point} for point $p$ if there exists
covering~$\mathcal H_1 \cup \mathcal H_2$ of~$\mathcal{T}[p, r(p)]$ so
that, assuming $p$ is on the boundary of $\mathcal H_1$:

\begin{itemize}[noitemsep]
\item The point $t$ lies on the boundary of both $\mathcal H_1$ and $\mathcal H_2$, and
\item The points $t$ and $p$ are on opposite sides of $\mathcal H_1$.
\end{itemize}
\begin{figure}[tb]
    \centering
    \includegraphics{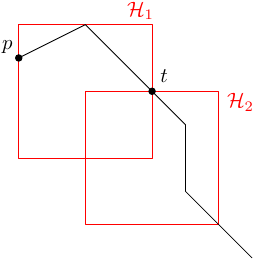}
    \caption{An example of a bridging point $t$ for the starting point $p$.}
    \label{fig:bridging_point}
\end{figure}

See Figure~\ref{fig:bridging_point} for an example. For a covering
$\H_1 \cup \H_2$ of $\T[p,r(p)]=\T[p,q]$, a side of $\H_1$ and $\H_2$ that
contains $p$ (respectively $q$) is a \emph{$p$-side} (respectively
\emph{$q$-side}). A side that contains a bridging point is a
\emph{$b$-side}. If two sides are part of the same square and have the
same orientation (vertical or horizontal), they are opposite to each other. 

\begin{observation}
  \label{obs:bridging_points}
  A bridging point $x$ either lies on both the upper and right
  envelopes of the subtrajectory, or on both the lower and left
  envelopes of the subtrajectory.
\end{observation}

\begin{lemma}
  \label{lem:bridging_lemma_part_1}
  Let $\mathcal{T}[p,q]$ be an optimal 2-coverable subtrajectory and
  assume that $p$ is not a vertex of $\mathcal{T}_2$. There is a
  covering of $\mathcal{T}[p,q]$ by squares
  $\mathcal H_1 \cup \mathcal H_2$ so that any side opposite to a
  $p$-side must either be a $b$-side or a $q$-side.
\end{lemma}

\begin{proof}
  Assume without loss of generality that $p$ lies on the left side of
  $\H_1$. We now argue that the right side of $\H_1$ is a $b$-side or
  a $q$-side.

  Since $p$ is not a vertex of $\mathcal{T}_2$, it lies between two
  consecutive vertices $s_i$ and $s_{i+1}$ of $\mathcal{T}_2$, as
  shown in Figure~\ref{fig:events_p_left_cases}.
  The segment $s_i s_{i+1}$ is a line segment. Consider the situation
  if we moved $\mathcal H_1$ to the left by an arbitrarily small
  amount. This would allow us to cover additional length of
  $s_i s_{i+1}$ on the left side of $\mathcal H_1$. Since
  $\mathcal{T}[p,q]$ is optimal, there must be a point
  $t \in \mathcal{T}[p,q]$ on the right side of $\mathcal H_1$ that
  does not lie in the interior of $\H_2$ which is lost even as we move left by an
  arbitrarily small amount. There are three possible cases for this
  point $t$, as shown in Figure~\ref{fig:events_p_left_cases}:

  \begin{figure}[t]
    \centering
    \includegraphics{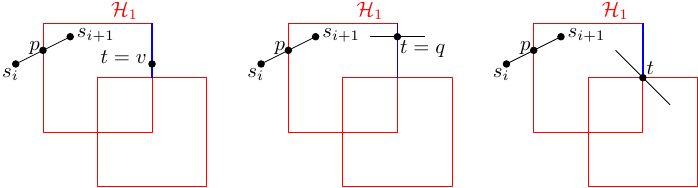}
    \caption{The three cases for the point $t$ that is lost when
      moving $\mathcal H_1$ to the left.}
    \label{fig:events_p_left_cases}
  \end{figure}

  \begin{description}
  \item[Case $t$ is a vertex $v$ of \T:] This would make $p$ an
    upper envelope event for vertex $v$. This would mean that $p$ is a
    vertex of $\mathcal{T}_2$, which contradicts the lemma
    statement. Hence, this case cannot actually occur.
  \item[Case $t$ is point $q$:] This means the right
    side of $\H_1$ is a $q$-side, as desired.
  \item[Case $t$ is an interior point of $\mathcal{T}$:] This means \T, in
    particular $\T[p,q]$, continues into $\H_2$ at $t$. This means $t$
    must lie on the boundary of $\H_2$, and thus the
    right side of $\H_1$ is a $b$-side, as desired.
  \end{description}

  Note that if $p$ lies on a corner of $\H_1$, say e.g. the top-left
  corner, there are two $p$-sides. Using the same argument as above,
  it then follows that the bottom-side of $\H_1$ must also either be a
  $q$-side or a $b$-side.
\end{proof}

\begin{lemma}
  \label{lem:bridging_lemma_part_2}
  Let $\mathcal{T}[p,q]$ be an optimal 2-coverable subtrajectory and
  assume that $p$ is not a vertex of $\mathcal{T}_2$. There is a
  covering of $\mathcal{T}[p,q]$ by squares
  $\mathcal H_1 \cup \mathcal H_2$ so that any side of $\H_2$ opposite
  to a $b$-side must be a $q$-side.
\end{lemma}

\begin{proof}
  Let $t$ be a bridging point that lies on the right side of $\H_1$,
  and the top side of $\H_2$. The other cases are symmetric. We then
  have to show that the bottom side of $\H_2$ is a $q$-side.

  To this end we move $\mathcal H_1$ left by an arbitrarily small
  amount and then move $\mathcal H_2$ upwards by an arbitrarily small
  amount to cover the neighborhood of the bridging point. By the same
  argument as in Lemma~\ref{lem:bridging_lemma_part_1}, there must be
  a point on $\mathcal{T}[p,q]$ that we lost on the bottom edge of
  $\mathcal H_2$. There are two cases as shown in
  Figure~\ref{fig:events_p_left_bridge_down_cases}.

  \begin{figure}[t]
    \centering
    \includegraphics{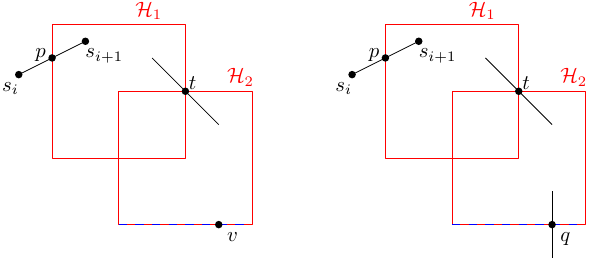}
    \caption{Cases when the point $p$ is not in a corner and there is a bridging point.}
    \label{fig:events_p_left_bridge_down_cases}
  \end{figure}

  The first case is if the point on the bottom edge of $\mathcal H_2$
  is a vertex. Refer to the left diagram in
  Figure~\ref{fig:events_p_left_bridge_down_cases}.
  This would make $p$ a bridge event and thus a vertex of
  $\mathcal{T}_2$, which would be a contradiction. Thus, the
  trajectory $\mathcal{T}$ must exit the covering
  $\mathcal H_1 \cup \mathcal H_2$ on the bottom edge. The point on
  the bottom edge is $q$ and so the bottom side of $\mathcal H_2$ is a
  $q$-side, as required.

  In case $t$ also lies on a second side of $\H_2$, say the left side,
  then we also have to argue that the right side of $\H_2$ is a
  $q$-side. Suppose that the right side of $\H_2$ does not contain
  $q$, then it must contain a vertex $w$ of $\T[p,r(p)]$ (otherwise we
  could again shift $\H_2$ left). However, since $p$ lies on the left
  side of $\H_1$ (by definition of $t$), this would make $p$ a reach
  event of $w$, and thus a vertex of $\T_2$. Contradiction.
\end{proof}

Next we use the above lemmas to show that either $p$ or $q$ is in a
corner of $\mathcal H_1$ or $\mathcal H_2$.

\begin{lemma}[The corner lemma]
  \label{lem:not_corner}
  Suppose $p$ is not a vertex of $\mathcal{T}_2$ and $\mathcal{T}[p,q]$ is optimal. Then we have that either $p$ is in a corner of $\mathcal H_1$ or $\mathcal H_2$ or that $q$ is in a corner of $\mathcal H_1$~or~$\mathcal H_2$.
\end{lemma}

\begin{proof}
We assume by contradiction that that $p$ is not a vertex of $\mathcal{T}_2$, and there is a covering $\mathcal H_1 \cup \mathcal H_2$ of the subtrajectory $\mathcal{T}[p,q]$ where neither of $p$ nor $q$ are in a corner of $\mathcal H_1$ or $\mathcal H_2$. We show that this implies that the subtrajectory $\mathcal{T}[p,q]$ cannot be optimal.

We have two cases. Either the point $p$ lies on the left side of $\mathcal H_1$ or the right side of $\mathcal H_1$. All other cases are symmetric.

\begin{description}
  \item[Case 1: $p$ is on the left side of $\mathcal H_1$.] The left side of $\mathcal H_1$ is a $p$-side, so by Lemma~\ref{lem:bridging_lemma_part_1}, the right side of $\mathcal H_1$ is either a $q$-side or a $b$-side. We consider two subcases.

\textbf{Case 1.1: $p$ is on the left side of $\mathcal H_1$ and $q$ is on the right side.} See Figure~\ref{fig:four_cases_p_left_q_right}. 
\begin{figure}[ht]
    \centering
    \includegraphics{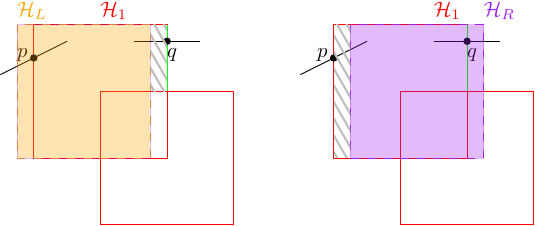}
    \caption{The new squares $\mathcal H_L$ and $\mathcal H_R$ are constructed to the left and right of $\mathcal H_1$ respectively.}
    \label{fig:four_cases_p_left_q_right}
\end{figure}

We show that $\mathcal{T}[p,q]$ cannot be optimal by constructing an earlier subtrajectory with the same length as $\mathcal{T}[p,q]$. We do this by constructing a subtrajectory covered by a square slightly to the left of $\mathcal H_1$, and another subtrajectory covered by a square slightly to the right of $\mathcal H_1$. 

If there is a vertex on the left side of $\mathcal H_1$, then 

that would make $p$ a bounding box event and a vertex of $\mathcal{T}_1$. If there is a vertex on the right side of $\mathcal H_1$, then 

that would make $p$ an upper envelope event and vertex of $\mathcal{T}_2$. Since $p$ is not a vertex of $\mathcal{T}_2$, we must have that there are no vertices on the left side or the right side of $\mathcal H_1$. 

Now, take $\mathcal H_1$ and move it left and right by the same arbitrarily small amount, and call these new squares $\mathcal H_L$ and $\mathcal H_R$ respectively. For $\mathcal H_L$ on the left diagram of Figure~\ref{fig:four_cases_p_left_q_right}, because there are no vertices or bridging points on the right edge, we can choose the movement small enough so that the gray diagonally shaded region is empty. We can do the same for $\mathcal H_R$ on the right diagram, because there are no vertices on the left edge, so we can choose the movement small enough so that the gray diagonally shaded region is also empty. So the only lengths gained or lost are those on the segment $e_p$ that contains $p$, or the segment $e_q$ that contains $q$.

Suppose that $\mathcal{T}[p,q]$ has length $L$ and by assumption is maximal. Let the length of trajectory we gain with $\mathcal H_L$ on segment $e_p$ be $\ell_p$ and the length of trajectory we lose on segment $e_q$ be $\ell_q$. By symmetry and the fact that $p$ and $q$ are not in corners of $\mathcal H_1$, we have that $\mathcal H_R$ loses the same amount $\ell_p$ and gains the same amount $\ell_q$. Therefore, we have trajectories close to $\mathcal{T}[p,q]$ with lengths $L$, $L-\ell_p+\ell_q$ and $L+\ell_p-\ell_q$ respectively. Since $L$ is maximal, we must have $\ell_p = \ell_q$. But now we have an earlier trajectory with the same length as $\mathcal{T}[p,q]$, so $\mathcal{T}[p,q]$ is not optimal.

\textbf{Case 1.2: $p$ is on the left side of $\mathcal H_1$ and there is a bridging point on the right side.} See Figure~\ref{fig:four_cases_p_left_q_down}.

\begin{figure}[ht]
    \centering
    \includegraphics{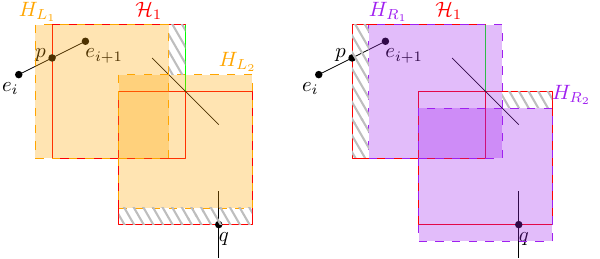}
    \caption{New hotspot positions for starting slightly before and after $p$.}
    \label{fig:four_cases_p_left_q_down}
\end{figure}

It follows from Observation~\ref{obs:bridging_points} that this
bridging point must lie on the top side of $\H_2$ (as it must lie on
the upper and right envelopes of $\T[p,r(p)]$). Then,
Lemma~\ref{lem:bridging_lemma_part_2} tells us that the bottom side of
$\H_2$ must be a $q$-side.

We now follow a similar shifting argument as before. This time we
shift both $\H_1$ and $\H_2$. We move $\mathcal H_1$ left to form
$\mathcal H_{L_1}$, and $\mathcal H_2$ up to form $\mathcal H_{L_2}$
and cover the part of the upper envelope left uncovered by
$\mathcal H_{L_1}$. We apply the opposite movements to obtain
$\mathcal H_{R_1}$ and $\mathcal H_{R_2}$ respectively.

Again, the gray regions on the left and right sides of $\mathcal H_1$ are empty if we choose the movement small enough, otherwise $p$ would be a vertex of $\mathcal{T}_2$. For $\mathcal H_2$ we have an analogous reason but this time it would make $q$ a vertex of $\mathcal{T}_2$. Therefore, the only parts of the trajectory gained or lost are those close to $p$ or $q$. By the same argument as before, these lengths $\ell_p$ and $\ell_q$ are the same for $\mathcal H_{L_1} \cup \mathcal H_{L_2}$ and $\mathcal H_{R_1} \cup \mathcal H_{R_2}$. Therefore, we again deduce that in order for $\mathcal{T}[p,q]$ to be optimal, $\ell_p$ and $\ell_q$ must be equal, but then $\mathcal{T}[p,q]$ is not the earliest optimal trajectory.
\item[Case 2: $p$ is on the right side of $\mathcal H_1$.] By
  Lemma~\ref{lem:bridging_lemma_part_1}, the left side of $\mathcal H_1$ is a
  $q$-side, since it cannot be a $b$-side (by
    Observation~\ref{obs:bridging_points}).
  The proof of this case is exactly the same as Case~1. We use the same construction as the one shown in Figure~\ref{fig:four_cases_p_left_q_right}.
\end{description}

In all cases, if $p$ is not a vertex of $\T_2$ and $p$ and $q$ are not
in corner positions of $\mathcal H_1$ and $\mathcal H_2$, then
$\mathcal{T}[p,q]$ is not an optimal subtrajectory.
\end{proof}

Finally, we show that $p$ is not only a corner of $\mathcal H_1$ or
$\mathcal H_2$, but also that $p$ is in fact a special configuration
event.

\begin{lemma}
\label{lem:all_events}
Suppose that~$\mathcal{T}[p,q]$ is optimal and that $p$ is not a vertex of $\mathcal{T}_2$. Then $p$ is a special configuration event.
\end{lemma}

\begin{proof}
By Lemma~\ref{lem:not_corner}, we must have $p$ or $q$ be in a corner of $\mathcal H_1$ or $\mathcal H_2$. Without loss of generality, suppose that $p$ is in a corner of $\mathcal H_1$. Up to rotation this leaves only two cases, either $p$ is in the top-left corner of $\mathcal H_1$, or $p$ is in the top-right corner of $\mathcal H_1$.

\textbf{Case 1: $p$ is in the top-left corner.} Consider the sides
opposite to $p$ on $\mathcal H_1$. These would be the bottom side and
right side of $\mathcal H_1$. Lemma~\ref{lem:bridging_lemma_part_1} implies
both these sides must contain either $q$ or a bridging point. Since
$\T[p,r(p)]$ enters $\H_2$, at least one of these two sides must
contain a bridging point.

Moreover, applying Lemma~\ref{lem:bridging_lemma_part_2} to the
bridging point implies that $q$ is in fact on the square
$\mathcal H_2$ and not $\mathcal H_1$. Therefore, both sides opposite
$p$ contain a bridging point, and $q$ is on $\mathcal H_2$. There are
two subcases. Either there are two different bridging points on the
bottom and right sides of $\mathcal H_1$, or there is a single
bridging point in the bottom-right corner of $\mathcal H_1$. See
Figure~\ref{fig:four_cases_p_top-left}.

\begin{figure}[ht]
    \centering
    \includegraphics{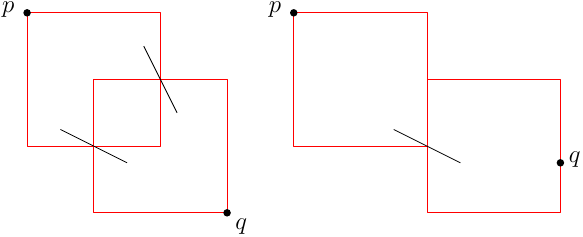}
    \caption{The special configuration events in Case 1.}
    \label{fig:four_cases_p_top-left}
\end{figure}

In the first subcase, there are two different bridging points on the bottom and right sides of $\mathcal H_1$, see Figure~\ref{fig:four_cases_p_top-left} left. Therefore, there are two bridging points, on the top and left edges of $\mathcal H_2$. We apply Lemma~\ref{lem:bridging_lemma_part_2} on the two bridging points. The bridging point on the top edge of $\mathcal H_2$ implies that $q$ must be on the bottom edge of $\mathcal H_2$, whereas the bridging point on the left edge implies that $q$ must be on the right edge of $\mathcal H_2$. So $q$ is in the bottom-right corner of $\mathcal H_2$. Therefore, $p$ is a special configuration event.

In the second subcase, there is a single bridging point in the bottom-right corner of $\mathcal H_1$. See Figure~\ref{fig:four_cases_p_top-right} right. This makes $p$ a special configuration event.

\textbf{Case 2: $p$ is in the top-right corner.}
Lemma~\ref{lem:bridging_lemma_part_1} implies that the left and bottom edges
of $\mathcal H_1$ must contain either $q$ or a bridging
point. However, by Observation~\ref{obs:bridging_points}, the left
edge of $\mathcal H_1$ cannot contain a bridging point, and thus $q$ lies on $\mathcal H_1$. Supposing there was a bridging point on
the bottom edge of $\mathcal H_1$,
Lemma~\ref{lem:bridging_lemma_part_2} would again imply that $q$ is on the right edge of $\mathcal H_2$, contradicting the fact that $q$ is on the left edge of $\mathcal H_1$. Therefore, there are no bridging points, $p$ is in the top-right corner of $\mathcal H_1$ and $q$ is on the bottom-left corner of $\mathcal H_1$. See Figure~\ref{fig:four_cases_p_top-right}. This makes $p$ a special configuration event.

\begin{figure}[ht]
    \centering
    \includegraphics{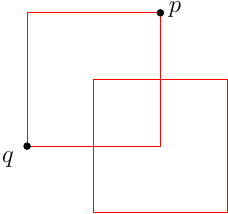}
    \caption{The special configuration event in Case 2.}
    \label{fig:four_cases_p_top-right}
\end{figure}
\end{proof}

By Lemma~\ref{lem:all_events}, the starting point of the optimal
subtrajectory is either a vertex of $\mathcal T_2$ or a special
configuration event. Hence, the starting point $p^*$ of the optimal
subtrajectory must be a vertex of $\mathcal T_3$, and thus
$p^* \in S$. We summarise with the following Theorem.

\begin{theorem}
  \label{thm:starting_point_guaranteed_to_be_in_t3}
  The set $S$ of vertices of $\mathcal T_3$ is guaranteed to contain
  the starting point of a longest coverable subtrajectory of $\T$.
\end{theorem}

\subsection{Computing the reach of a point}
\label{sub:reach_of_a_point}

In this section we describe how, given a candidate starting point $p$,
we can compute the longest 2-coverable subtrajectory starting at
$p$. We modify the data structure in Theorem~\ref{thm:problem2_k=2},
i.e. the data structure for answering whether a given subtrajectory is
2-coverable, to answer such reach queries. We do so by applying
parametric search to the query procedure. Note that applying a simple
binary search will give us only the edge containing
$r(p)$. Furthermore, even given this edge it is unclear how to find
$r(p)$ itself, as the squares may still shift, depending on the
exact position of $r(p)$.

\begin{lemma}
  \label{lem:reach_data_structure}
  Let $\T$ be a trajectory with $n$ vertices. After $O(n\alpha(n)\log n)$
  preprocessing time, $\T$ can be stored using $O(n\alpha(n)\log n)$ space, so
  that given a query point $p$ on $\T$ it can compute the reach $r(p)$
  of $p$ in $O(\log^2 n)$ time.
\end{lemma}

\begin{proof}
  We would like to compute the maximum value $q \in [0,1]$ so that
  $\T[p,q]$ is 2-coverable. The decision version of this optimisation
  problem is to decide whether a given subtrajectory $\T[p,q]$ is
  2-coverable. This decision version is monotone since for any
  $q'$ with $q \prec q'$, the subtrajectory $\T[p,q']$ contains the subtrajectory
  $\T[p,q]$. After $O(n\alpha(n) \log n)$ preprocessing time
  Theorem~\ref{thm:problem2_k=2} gives us a comparison-based
  algorithm, the query procedure, that solves the decision problem in
  $O(\log n)$ time. The sequential version of parametric
  search~\cite{megiddo1981applying} states that if $T$ is the running
  time of the sequential algorithm, the optimisation algorithm takes
  $O(T^2)$ time. In our case, the reach can be answered in
  $O(\log^2 n)$ time as required.
\end{proof}

\begin{corollary}
  \label{cor:compute_all_reaches}
  Given a trajectory $\T$, and a set of $m$ candidate starting points on
  $\T$, we can compute the longest 2-coverable subtrajectory that starts
  at one of those points in $O(n\alpha(n)\log n + m\log^2 n)$
  time.
\end{corollary}

\subsection{Computing the set of starting points}
\label{sub:set_of_starting_points}

Next, we bound the number of events, and thus the number of candidate starting
points. We also provide algorithms for computing the events, and we analyze the running times of our algorithms. Combining this with our
result from Corollary~\ref{cor:compute_all_reaches} gives us an
efficient algorithm to compute the optimal 2-coverable
subtrajectory. Section~\ref{apx:reach_f} is dedicated to reach events, Section~\ref{apx:bounding_box_f} to bounding box events, Section~\ref{apx:bridge_f} to bridge events, Section~\ref{apx:upper_envelope_f} to upper envelope events, and Section~\ref{apx:special_configuration_f} to special configuration events. 

\subsubsection{Reach events}
\label{apx:reach_f}

\begin{lemma}
Given a trajectory with $n$ vertices, there are at most $O(n)$ reach events which can be computed in $O(n \log^2 n)$ time.
\end{lemma}
\begin{proof}
Suppose $p$ is a reach event and $r(p)$ is a vertex. The vertex $r(p)$ uniquely defines $p$ since it is the earliest point on the trajectory $\mathcal{T}$ that reaches $r(p)$. Since there are $n$ vertices, there are at most $O(n)$ reach events.

The running time is immediately implied by Corollary~\ref{cor:compute_all_reaches}, as we are computing the reaches of all the vertices.
\end{proof}

\subsubsection{Bounding box events}
\label{apx:bounding_box_f}

\begin{lemma} \label{lem:bounding_box_f}
Given a trajectory with $n$ vertices, there are at most $O(n)$ bounding box events.
\end{lemma}
\begin{proof}
Suppose $p$ is such a bounding box event. Let the first and last vertices of $\mathcal{T}$ in the subtrajectory $\mathcal{T}[p, r(p)]$ be~$v_i$ and~$v_j$. We prove that the pair of vertices~$v_i$ and~$v_j$ uniquely determines $p$. Then we prove that there are at most $O(n)$ possible choices of the pair $v_i$ and $v_j$.

Suppose~$v_i$ and~$v_j$ are given. Let~$v_{i-1}$ be the vertex preceding~$v_i$, then $p$ must lie on the segment~$v_{i-1}v_i$. The vertex~$v_L$ is the unique leftmost vertex between~$v_i$ and~$v_j$. Now,~$v_L$ determines the $x$-coordinate of $p$, and since $p$ lies on ~$v_{i-1}v_i$, we have the unique position for $p$. Therefore, the vertices~$v_i$ and~$v_j$ uniquely determine $p$.

Analogous to in Section~\ref{sec:A_Longest_1-coverable_subtrajectory}
there are $O(n)$ relevant pairs of vertices $v_i$, $v_j$ that we have
to consider, and each pair $(v_i, v_j)$ uniquely determines the
bounding box event $p$, we have that there are at most $O(n)$ bounding
box events.
\end{proof}

\begin{lemma} \label{lem:bounding_box_g}
Given a trajectory with $n$ vertices, one can compute all bounding box events in $O(n \log^2 n)$ time.
\end{lemma}
\begin{proof}
From the proof in Lemma~\ref{lem:bounding_box_f}, we know that $p$ can be determined by the pair $(v_i, v_j)$. We also proved the relationship between the pair $(v_i, v_j)$ and the longest coverable vertex-to-vertex subtrajectory starting at either~$v_i$ or ending at~$v_j$. As a consequence of Lemma~\ref{lem:reach_data_structure}, we can compute all longest coverable vertex-to-vertex subtrajectories starting at each vertex~$v_i$ in $O(n \log^2 n)$ time. Those ending at vertex~$v_j$ can be handled analogously. 

For each of the $O(n)$ pairs of vertices $(v_i, v_j)$ we use the same method as in Lemma~\ref{lem:bounding_box_f} to determine the bounding box event~$p$. We use the bounding box data structure to query~$v_L$ in $O(\log n)$ time. This determines the $x$-coordinate of~$p$. Then we compute~$p$ by computing the intersection of the two lines: the vertical line through~~$v_L$ and the trajectory edge~$v_{i-1}v_i$.
\end{proof}

\subsubsection{Bridge events}
\label{apx:bridge_f}

\begin{lemma}
Given a trajectory with $n$ vertices, there are at most $O(n)$ bridge events.
\end{lemma}
\begin{proof}
Let the first and last vertices in the subtrajectory $\mathcal{T}[p, r(p)]$ be~$v_i$ and~$v_j$. By Lemma~\ref{lem:bounding_box_f},

there are $O(n)$ relevant pairs of vertices $(v_i, v_j)$. It suffices to show that for each pair $(v_i, v_j)$ there are only a constant number of bridge events. The pair $(v_i, v_j)$ determines the leftmost vertex~$v_L$ on $\mathcal{T}[v_i, v_j]$. If $x$ is in $\mathcal{T}[v_i, v_j]$ then it is the unique point on the upper envelope of $\mathcal{T}[v_i, v_j]$ one unit to the right of~$v_L$. Otherwise, $x$ is on~$v_{i-1}v_i$ or~$v_j v_{j+1}$ and one unit to the right of~$v_L$. There are at most two possible positions for the bridging point $x$. Thus there are a constant number of bridge events for each of the $O(n)$ pairs $(v_i, v_j)$.
\end{proof}

\begin{lemma}
Given a trajectory with $n$ vertices, one can compute all bridge events in $O(n \log^2 n)$ time.
\end{lemma}
\begin{proof}
In a similar manner to the proof of Lemma~\ref{lem:bounding_box_g}, we
begin by computing all pairs $(v_i, v_j)$  in $O(n \log^2 n)$
time. For each pair $(v_i, v_j)$ we compute the vertex~$v_L$ in
$O(\log n)$ time with the bounding box data structure. Consider two
cases. If $x$ is in $\mathcal{T}[v_i, v_j]$ we query the upper
envelope of $\mathcal{T}[v_i,v_j]$ in $O(\log n)$ time with the upper
envelope data structure. Otherwise, if $x$ is not in $\mathcal{T}[v_i,
v_j]$, then $x$ is on~$v_{i-1}v_i$ or~$v_j v_{j+1}$ and we can compute
the intersection in $O(1)$ time using $\BB(\T[v_i,v_j])$. Therefore, the running time is $O(n \log^2 n)$ time in total.
\end{proof}

\subsubsection{Upper envelope events}
\label{apx:upper_envelope_f}

\begin{lemma} \label{lem:upper_envelope_f}
Given a trajectory with $n$ vertices, there are $O(n 2^{\alpha(n)})$ upper envelope events of $\mathcal{T}$, where~$\alpha$ is the inverse Ackermann function.
\end{lemma}
\begin{proof}
For each upper envelope event $p$ of the trajectory $\mathcal{T}$, let $u(p)$ be the segment of $\mathcal{T}$ on the upper envelope of $\mathcal{T}[p,r(p)]$ that is one unit to the right of~$p$. If there are multiple such segments, take any of them. As $p$ ranges from the earliest upper envelope event to the last one, $u(p)$ is a sequence of segments. It suffices to show that $u(p)$ is bounded from above by $O(n 2^{\alpha(n)})$. We achieve this by showing that the sequence of segments $u(p)$ is a Davenport-Schinzel sequence of order $s=4$~\cite{DBLP:books/daglib/davenportschinzel}.

Recall that a Davenport-Schinzel sequence of order $s=4$ has no alternating subsequences of length $s+2 = 6$. The subsequence cannot occur anywhere in the sequence even for non-consecutive appearance of the terms. Our first step is to show that if the sequence $a,b,a$ occurs (not necessarily consecutively) then the first two elements of the sequence must be $x$-monotone, in that the first element is to the left of the second element. Our second step is to deduce a contradiction from an alternating and $x$-monotone subsequence of length five.

Suppose that $a,b,a$ is a subsequence of $u(p)$, then there exists three upper envelope events~$p_1 \prec p_2 \prec p_3$ along the trajectory~$\mathcal{T}$ so that $u(p_1), u(p_2), u(p_3) = a,b,a$. In other words, segment $a = u(p_1) = u(p_3)$ whereas segment $b = u(p_2)$. Suppose for the sake of contradiction that~$p_2$ is to the left of~$p_1$. See Figure~\ref{fig:davenport_schinzel}.

\begin{figure}[ht]
    \centering
    \includegraphics{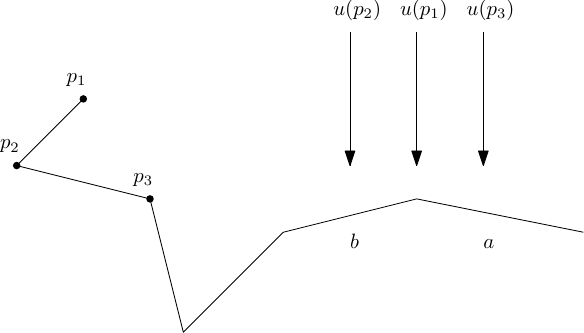}
    \caption{Three upper envelope events $p_1, p_2, p_3$ so that $u(p_1), u(p_2), u(p_3)$ are alternating.}
    \label{fig:davenport_schinzel}
\end{figure}

Recall that since~$p_1$ is an upper envelope event, $p_1$~is the leftmost point of $\mathcal{T}[p_1, r(p_1)]$. But~$p_2$ is to the left of $p_1$, so we must have that~$p_2 \not \in \mathcal{T}[p_1, r(p_1)]$, and therefore $r(p_1) \prec p_2$. Moreover, for any point $p$, we have $p \prec u(p) \prec r(p)$. Combining these, we get:
$$ u(p_1) \prec r(p_1) \prec p_2 \prec p_3 \prec u(p_3) = u(p_1).$$

This is a contradiction, so~$p_1$ is to the left of~$p_2$. Therefore, whenever the alternating subsequence $u(p_1), u(p_2), u(p_3) = a,b,a$ occurs, the first two elements~$p_1$ and~$p_2$ are $x$-monotone.

Now suppose we have an alternating subsequence $a,b,a,b,a,b$ of length 6. Let the subsequence be $u(p_1), u(p_2), u(p_3), g(p_4), g(p_5), g(p_6)$. By the property above, we have that $p_1, p_2, p_3, p_4$ and $p_5$ are $x$-monotone. Since $u(p_1) = g(p_5) = a$, the segment $a$ spans the entire $x$-interval from $u(p_1)$ to $g(p_5)$. But now $u(p_2) = g(p_4) = b$, which means that segment $b$ is above segment $a$ at $u(p_2)$ and $g(p_4)$. Since $a$ and $b$ are straight, this implies that $b$ is also above $a$ at $u(p_3)$. But $u(p_3) = a$, which is a contradiction. Therefore the alternating subsequence of length 6 does not occur and $u(p)$ is a Davenport-Schinzel sequence of order $s=4$.
\end{proof}

\begin{lemma} \label{lem:upper_envelope_g}
  Given a trajectory with $n$ vertices, one can compute all upper
  envelope events in $O(n\alpha(n) \log^2 n)$ time.
\end{lemma}
\begin{proof}
We begin with a preprocessing step. We compute a set $S$ of all vertex events, reach events, and bounding box events of $\mathcal{T}$. Since there is one reach event per vertex there are $O(n)$ reach events, and combined with Lemma~\ref{lem:upper_envelope_f},
this means that $S$ has size $O(n)$. 

The set $S$ has three properties. The first property is that between any two consecutive events~$s_i$ and~$s_{i+1}$, the trajectory $\mathcal{T}$ is a straight segment, since all vertices of $\mathcal{T}$ are in $S$. The second property is that for the set of points $p \in \T[s_i, s_{i+1}]$, their set of reaches $\{r(p): p \in \T[s_i, s_{i+1}]\}$ must lie on a straight segment of $\mathcal{T}$. The reason for this is that if there were a vertex strictly between $r(s_i)$ and $r(s_{i+1})$, then there would be a (reach) event between~$s_i$ and~$s_{i+1}$, contradicting the fact that~$s_i$ and~$s_{i+1}$ are consecutive. Finally, the third property is that, supposing~$s_i$ is the leftmost point on the subtrajectory $\mathcal{T}[s_i, r(s_i)]$, then for any $p \in \T[s_i, s_{i+1}]$, $p$ is the leftmost point on the subtrajectory $\mathcal{T}[p, r(p)]$. The reason for this is that if there were $p$ that had a vertex $v$ to the left of $p$, then there would be a bounding box event strictly between~$s_i$ and~$s_{i+1}$.

Next, we extend these properties of $S$ to properties of upper envelope events that are between~$s_i$ and~$s_{i+1}$. Let~$v_i$ and~$v_j$ be the first and last vertices of $\mathcal{T}[p, r(p)]$ for some point $p \in \T[s_i, s_{i+1}]$. As a consequence of the first two properties of set $S$, the vertices~$v_i$ and~$v_j$ are the same regardless of our choice of point $p$. Now suppose that $p$ is an upper envelope event. This means that $p$ is to the left of all vertices on the subtrajectory $\mathcal{T}[v_i, v_j]$. As a consequence of the third property of set $S$, both~$s_i$ and~$s_{i+1}$ have $x$-coordinate less than or equal to the $x$-coordinate of all vertices of the subtrajectory $\mathcal{T}[v_i, v_j]$.

Now the algorithm is to take each pair of consecutive events $(s_i, s_{i+1})$ and compute the upper envelope events that occur between~$s_i$ and~$s_{i+1}$. We decide on a subset of these pairs $(s_i, s_{i+1})$ to skip, since they will have no upper envelope events. For each pair of consecutive events $(s_i, s_{i+1})$, compute the vertices~$v_i$ and~$v_j$ (which are the first and last vertices of $\mathcal{T}[p, r(p)]$ for any point $p \in \T[s_i, s_{i+1}]$). From the definition of an upper envelope event we have that  the first requirement on an upper envelope $p$ implies that $p$ is to the left of the entire subtrajectory $\mathcal{T}[v_i, v_j]$. This implies that if the segment $s_is_{i+1}$ is not entirely to the left of $\mathcal{T}[v_i, v_j]$, we can skip the pair $(s_i, s_{i+1})$. 

The second requirement is that $p$ is one unit to the right of an inflection point $u$ on the upper envelope of $\mathcal{T}[v_i, v_j]$. In particular, if $x_i$ and $x_{i+1}$ are the $x$-coordinates of~$s_i$ and~$s_{i+1}$, then computing the upper envelope events $p$ in the vertical strip $[x_i, x_{i+1}]$ is equivalent to computing the inflection points $u$ in the vertical strip $V=[x_i+1, x_{i+1}+1]$.

Our problem is now to compute the upper envelope of $\mathcal{T}[v_i, v_j]$ in the vertical strip~$V$. For each of the $O(\log n)$ canonical subsets of the subtrajectory $\mathcal{T}[v_p, v_q]$, we compute the upper envelope $\Gamma_i$ for that canonical subset. The upper envelope of $\mathcal{T}[v_i, v_j]$ is simply the upper envelope of the $O(\log n)$ upper envelopes $\Gamma_i$. In order to argue amortised complexity for computing the upper envelope of the $\Gamma_i$'s, we proceed with a sweepline algorithm. 

Suppose our vertical sweepline is~$\ell$. Let its initial state $\ell_{start}$ be the left boundary of~$V$, and its ending state $\ell_{end}$ be the right boundary of~$V$. We maintain three invariants for the sweepline~$\ell$. First, we maintain pointers~$p_i$ to mark the positions and directions of each of the $\Gamma_i$. Second, we maintain the current highest of the pointers~$p_i$, which we will call~$p_{max}$. Finally, we maintain possible intersections where~$p_{max}$ may change, as such we maintain the intersection of~$p_{max}$ with each other~$p_i$. See Figure~\ref{fig:gamma_i}.

\begin{figure}[ht]
    \centering
    \includegraphics{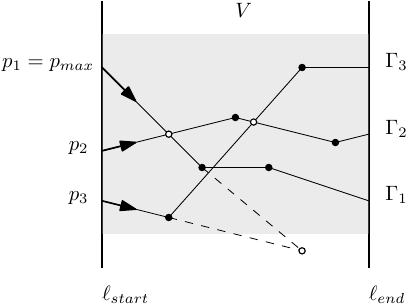}
    \caption{Sweepline maintains pointers $p_i$. Solid dots: $p_i$ changes. Hollow dots: $p_{max}$ swaps.}
    \label{fig:gamma_i}
\end{figure}

There are two types of sweepline events. The first type of sweepline event occurs when a pointer~$p_i$ changes. These sweepline events are marked with solid dots in Figure~\ref{fig:gamma_i}, and are the inflection points of $\Gamma_i$. In our update step, we update the pointer~$p_i$ and the intersection(s) between~$p_{max}$ and~$p_i$. The second type of sweepline event occurs~$p_{max}$ changes, in particular when it swaps with some other pointer~$p_i$. These intersection points are marked with hollow dots in Figure~\ref{fig:gamma_i}. In our update step, we update~$p_{max}$ and all intersections between~$p_{max}$ and~$p_i$.

Once the sweepline algorithm terminates, the segments traced by the pointer $p_{max}$ corresponds to the upper envelope of~$\mathcal{T}[v_i, v_j]$. We compute the inflection points along~$p_{max}$ and our algorithm returns all upper envelope events $p$ on~$s_i s_{i+1}$ which are one unit to the left of an inflection point.

It remains to analyse the amortised running time of this algorithm. By Corollary~\ref{cor:compute_all_reaches} we can compute a reach event for each vertex in $O(n \log^2 n)$ time. By Lemma~\ref{lem:bounding_box_g} we can compute all bounding box events in $O(n \log^2 n)$. We construct the upper envelope of all canonical subsets of $\mathcal{T}$ in $O(n \log^2 n)$ time~\cite{DBLP:journals/ipl/Hershberger89}. We initialise the sweepline algorithm and compute all $O(\log n)$ pointers in $O(\log^2 n)$ time. When the direction of a pointer changes, we update the pointer in constant time, and calculate the new intersections between~$p_{max}$ and~$p_i$. Since each new intersection can be computed in constant time, and there are $O(\log n)$ intersections to calculate, this step takes $O(\log n)$ time. When the highest pointer~$p_{max}$ changes, we update~$p_{max}$ in constant time, and calculate new intersections in $O(\log n)$ time. Therefore, the amortised running time of the sweepline algorithm is $O(\log n)$ per sweepline event. Hence, it suffices to count the number of sweepline events.

The first type of sweepline event is when the direction of the pointer~$p_i$ changes. The number of times a pointer~$p_i$ changes is equal to the number of inflection points of $\Gamma_i$ in the vertical strip~$V$. Suppose that we charge the sweepline event to that inflection point on~$\Gamma_i$. If we show that each inflection point on~$\Gamma_i$ gets charged at most once, not just during a single sweepline algorithm but in total across all pairs $(s_i, s_{i+1})$, then the total number of sweepline events of this type is bounded by the total complexity of all the $\Gamma_i$'s. The total complexity of all upper envelopes of all canonical subsets of the trajectory is $O(n \alpha(n) \log n)$~\cite{DBLP:journals/ipl/Hershberger89}.

Suppose for a sake of contradiction that two sweepline events charge to the same inflection point $u$. Since the sweepline algorithm sweeps from left to right without backtracking, these two sweepline events must have originated from two different pairs. 

\begin{figure}[ht]
    \centering
    \includegraphics{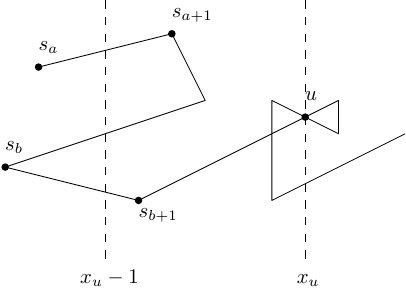}
    \caption{If~$s_b$ is after~$s_{a+1}$, then~$s_b$ contradicts the leftmost point property of~$s_{a+1}$.}
    \label{fig:double_charging_u}
\end{figure}

Suppose that the inflection point $u$ is charged by sweepline events originating from pairs $(s_a, s_{a+1})$ and $(s_b, s_{b+1})$. Without loss of generality let~$s_a \prec s_{a+1} \prec s_b$ along the trajectory~$\mathcal{T}$. Refer to Figure~\ref{fig:double_charging_u}. We will show that this contradicts the third property of the set $S$, which states that~$s_{a+1}$ is the leftmost point on the subtrajectory $\mathcal{T}[s_{a+1}, r(s_{a+1})]$. To this end we will show that~$s_b$ is between~$s_{a+1}$ and $r(s_{a+1})$ along the trajectory $\mathcal{T}$, and that~$s_b$ is to the left of~$s_{a+1}$.

Note that $u$ occurs as a sweepline event for~$s_b$ so~$s_{a+1} \prec s_b \prec u$. Therefore~$s_{a+1} \prec s_b \prec u \prec r(s_{a+1})$. It remains to show that~$s_b$ is to the left of~$s_{a+1}$. Let the $x$-coordinate of $u$ be $x_u$ and consider the vertical line at $x$-coordinate $x_u - 1$, one unit to the left of $u$. Since both sweepline algorithms for $s_a$ and $s_b$ visited the inflection point $u$, we must have that the vertical line cuts~$s_a s_{a+1}$ and~$s_b s_{b+1}$ in such a way that $s_a$ and $s_b$ are to the left of the vertical line, whereas $s_{a+1}$ and $s_{b+1}$ are to the right of the vertical line. Therefore,~$s_b$ is to the left of~$s_{a+1}$, completing our proof by contradiction. Hence, no inflection point $u$ can be charged twice for the first type of sweepline event. 

The second type of sweepline event is when the highest pointer~$p_{max}$ changes. Every time the second type of sweepline event occurs, there is a new upper envelope event. Therefore, the number of events of the second type is bounded by the number of upper envelope events, which by Lemma~\ref{lem:upper_envelope_f}
is at most~$O(n 2^{\alpha(n)})$. Therefore, the number of sweepline events is dominated by the first type. 

The total running time of the sweepline algorithm is $O(\log n)$ time per sweepline event, which leads to $O(n \alpha(n) \log^2 n)$ time in total. Therefore, the overall running time of this algorithm $O(n \alpha(n) \log^2 n)$.
\end{proof}

\subsubsection{Special configuration events}
\label{apx:special_configuration_f}

We start by proving useful properties of consecutive vertices of $\T_2$.

\begin{lemma}
  \label{prop:2}
  Let $s_i$ and $s_{i+1}$ be a pair of consecutive vertices of
  $\mathcal{T}_2$. For all points $p \in \T[s_i, s_{i+1}]$, their set of
  reaches lie on a single edge $e$ of \T, i.e. $\{r(p) \mid p \in
  \T[s_i, s_{i+1}]\} \subseteq e$.
\end{lemma}

\begin{proof}
  Suppose there were a vertex $v$ in the set
  $\{r(p) \mid p \in \T[s_i, s_{i+1}]\}$. Then the point $p$ such that
  $r(p) = v$ would be a reach event, which would contradict that fact
  that $s_i$ and $s_{i+1}$ are consecutive vertices of
  $\mathcal{T}_2$, and hence all points in
  $\{r(p) \mid p \in \T[s_i, s_{i+1}]\}$ lie on an edge of \T.
\end{proof}

\begin{lemma}
  \label{prop:3}
  Let $s_i$ and $s_{i+1}$ be consecutive vertices of $\mathcal{T}_2$,
  and for any $p \in \T[s_i, s_{i+1}]$ let $u(p)$ be the point on the
  upper envelope of $\mathcal{T}[p, r(p)]$ that is one unit to the
  right of $p$. The set of points $\{u(p) \mid p \in \T[s_i,
  s_{i+1}]\}$ lies on a single edge of \T.
\end{lemma}

\begin{figure}[ht]
    \centering
    \includegraphics[width=\textwidth]{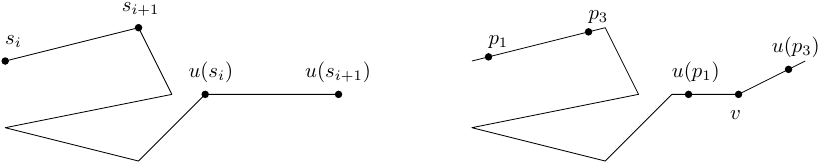}
    \caption{(Left) The consecutive vertices $s_i$ and $s_{i+1}$, and the points $u(s_i)$ and $u(s_{i+1})$ one unit to their right. (Right) If $u(p_1)$ and $u(p_3)$ are on different edges, then $p_1$ and $p_3$ are not consecutive.}
    \label{fig:lemma16_figure}
\end{figure}

\begin{proof}
  Assume for sake of contradiction that there are two points
  $p_1,p_3 \in \mathcal T[s_i, s_{i+    1}]$ for which $u(p_1)$ and $u(p_3)$ lie
  on different edges of $\mathcal T$. See Figure~\ref{fig:lemma16_figure}, (right). There must be a vertex $v$ of the
  upper envelope separating $u(p_1)$ and $u(p_3)$. As a result, it follows
  that there is a point $p_2 \in \mathcal T[p_1,p_3]$ for which
  $u(p_2)=v$. But now $p_2$ is a vertex of $\mathcal T_2$. This contradicts the fact that $p_1$ and $p_3$ are consecutive. Therefore, all $u(p)$ for $p \in
  \mathcal T[s_1,s_2]$ lie on a single edge of $\mathcal T$.
\end{proof}

\begin{lemma} \label{lem:special_configuration_f}
Given a trajectory with $n$ vertices, there are at most $O(n 2^{\alpha(n)})$ special configuration events.
\end{lemma}
\begin{proof}
We show the bound by showing that between any two elements of $\mathcal{T}_2$, there is either a unique special configuration event, or if there are multiple they are equivalent and we need only compute one of them. We show this by using Lemmas~\ref{prop:2} and~\ref{prop:3} of the trajectory $\mathcal{T}_2$. We require a rotated version of Lemma~\ref{prop:3} to hold for the left cardinal direction as well as the upward cardinal direction to bound the number of occurrences of special configuration~3. Now we consider three cases. 

\textbf{Special Configuration 1.} We use Lemma~\ref{prop:2}
of $\mathcal{T}_2$ to bound the number of special configuration events. We show that between consecutive events $s_i$ and $s_{i+1}$, there is either a unique instance of special configuration 1, or there are multiple instances of special configuration 1 which are all equivalent and we only need to compute one of them.

Let $e_p$ be the edge of $\mathcal{T}$ containing $p$ and let $e_q$ be the edge containing $q$. The segment $s_i s_{i+1}$ is a subset of $e_p$, and by Lemma~\ref{prop:2} of $\mathcal{T}_2$, the set of reaches $\{r(p): p \in \T[s_i, s_{i+1}]\}$ is a subset of $e_q$. Special configuration 1 states that the top-right corner of $\mathcal H_1$ lies on $e_p$ and the bottom-left corner of $\mathcal H_1$ lies on $e_q$. 

Let $p(t)$ be a function that slides the starting point $p$ from $s_i$ to $s_{i+1}$. Formally, let $p: [0,1] \to \T[s_i, s_{i+1}]$ be a linear function so that $p(0) = s_i$ and $p(1) = s_{i+1}$. Let $\mathcal H(t)$ be the unit sized square with its top-right corner at $p(t) \in e_p$. See Figure~\ref{fig:special3}. If $p(t)$ is in special configuration 1, then $\mathcal H(t)$ would also have its bottom-left corner on $e_q$. 

\begin{figure}[ht]
    \centering
    \includegraphics{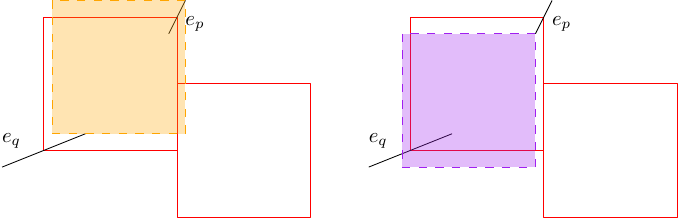}
    \caption{The square $\mathcal H(t)$ with top-right corner $p(t)$ on segment $e_p$.}
    \label{fig:special3}
\end{figure}

There are two cases. In the first case, $e_p$ and $e_q$ are not parallel. Then since $\mathcal H(t)$ moves parallel to $p(t) = e_p$, the bottom-left corner of $\mathcal H(t)$ moves parallel to $e_p$ with time. Therefore, the bottom-left corner of $\mathcal H(t)$ can only intersect $e_q$ once, and we have that between $s_i$ and $s_{i+1}$ there is a unique instance of special configuration 1.

Otherwise, $e_p$ and $e_q$ are parallel. Therefore, if it is true that $\mathcal H(t)$ has its bottom-left corner on $e_q$ for some value of $t$, then it is true for all values of $t \in [0,1]$. Moreover, $p(t)$ and $r(p(t))$ move along $e_p$ and $e_q$ at the same rate since they are opposite corners of a fixed sized square. We can deduce that $\mathcal{T}[p(t), r(p(t))]$ have the same length for all $t \in [0,1]$, in which case we only need to compute one such $p(t)$.

\textbf{Special Configuration 2.} We use Lemma~\ref{prop:3}
of $\mathcal{T}_2$ to show that it suffices to consider a unique instance of special configuration 2 between $s_i$ and $s_{i+1}$. Let $e_p$ be the segment of $\mathcal{T}$ containing $p$ and passing through the top-left corner of $\mathcal H_1$. Let $e_b$ be the segment of $\mathcal{T}$ that passes through the bottom-right corner of $\mathcal H_1$. 

The segment $s_i s_{i+1}$ is a subset of $e_p$. By Lemma~\ref{prop:3},
the set of points $\{u(p): p \in \T[s_i, s_{i+1}]\}$ is a subset of $e_b$. Let $p: [0,1] \to \T[s_i, s_{i+1}]$ be a linear function so that $p(0) = s_i$ and $p(1) = s_{i+1}$. Let $\mathcal H(t)$ be the unit sized square with its top-left corner at $p(t) \in e_p$. See Figure~\ref{fig:special2}.

\begin{figure}[ht]
    \centering
    \includegraphics{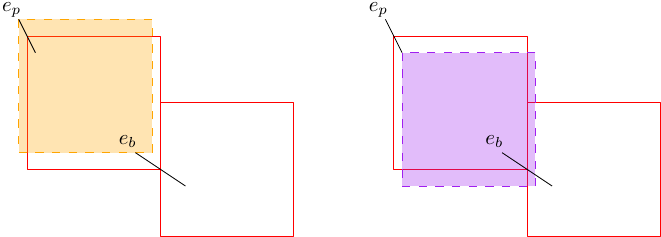}
    \caption{The square $\mathcal H(t)$ with top-left corner $p(t)$ on the segment $e_p$.}
    \label{fig:special2}
\end{figure}

If $p(t)$ were a special configuration event, then the bottom-right corner of $\mathcal H(t)$ is required to be on $e_b$. By the same reasoning as in special configuration 1, if $e_p$ and $e_b$ are not parallel, then there is at most one value of $t$ where this can hold. If $e_p$ and $e_b$ are parallel, then computing any candidate would suffice. Hence it suffices to consider a unique instance of special configuration 2 between $s_i$ and $s_{i+1}$.

\textbf{Special Configuration 3.} We use Lemmas~\ref{prop:2} and~\ref{prop:3}
of $\mathcal{T}_2$ to show that it suffices to consider a unique instance of special configuration 3 between $s_i$ and $s_{i+1}$. We use the Lemma~\ref{prop:3} 
for both the upward and leftward cardinal directions. 

Let $e_p$ be the segment of $\mathcal{T}$ that contains $p$ and passes through the top-left corner of $\mathcal H_1$. Let $e_q$ be the segment of $\mathcal{T}$ that contains $q$ and passes through the bottom-right corner of $\mathcal H_2$. Of the two distinct intersections of $\mathcal H_1$ and $\mathcal H_2$, let $e_{b_1}$ be the segment of $\mathcal{T}$ that passes through the intersection of the left edge of $\mathcal H_1$ with the top edge of $\mathcal H_2$, and let $e_{b_2}$ be the segment of $\mathcal{T}$ that passes through the intersection of the bottom edge of $\mathcal H_1$ with the left edge of $\mathcal H_2$. See Figure~\ref{fig:special1}.

Note that $s_i s_{i+1}$ is a subset of $e_p$. By Lemma~\ref{prop:2},
$\{r(p): p \in \T[s_i, s_{i+1}]\}$ is a subset of $e_q$. By Lemma~\ref{prop:3}
in the upward direction, $\{u(p): p \in \T[s_i, s_{i+1}]\}$ is a subset of $e_{b_1}$. If $l(p)$ is the point one unit below $p$ on the left envelope of $\mathcal{T}[p, r(p)]$, then by Lemma~\ref{prop:3} in the left direction, $\{l(p): p \in \T[s_i, s_{i+1}]\}$ is a subset of $e_{b_1}$.

Now let $p: [0,1] \to \T[s_i, s_{i+1}]$ be a linear function so that $p(0) = s_i$ and $p(1) = s_{i+1}$. Let $\mathcal H_1(t)$ be the unit sized square with its top-left corner at $p(t) \in e_p$. By the definition of an upper envelope event, 
$u(p(t))$ is one unit to the right of $p(t)$, and therefore $u(p(t))$ is the intersection of the right edge of $\mathcal H_1(t)$ and $e_{b_1}$. Similarly, $l(p(t))$ is the intersection of the bottom edge of $\mathcal H_1(t)$ and $e_{b_2}$. 

If $p(t)$ were a special configuration event, then there would exist a square $\mathcal H_2$ so that $u(p(t))$ is on the top edge of $\mathcal H_2$, $l(p(t))$ is on the left edge of $\mathcal H_2$, and $r(p(t))$ is in the bottom right corner of $\mathcal H_2$. Define $\mathcal H_2(t)$ to be the square with $u(p(t))$ on its top edge and $l(p(t))$ on its left edge. Then as $t$ varies linearly, $\mathcal H_1(t)$ moves linearly in the plane and therefore $u(p(t))$ and $l(p(t))$ move linearly along the segments $e_{b_1}$ and $e_{b_2}$. See Figure~\ref{fig:special1}. Therefore, $\mathcal H_2(t)$ moves linearly in the plane. For the same reason as in special configuration 1 and 2, it suffices to consider a unique position where the bottom-right corner of $\mathcal H_2(t)$ is on the segment $e_q$.

\begin{figure}[ht]
    \centering
    \includegraphics{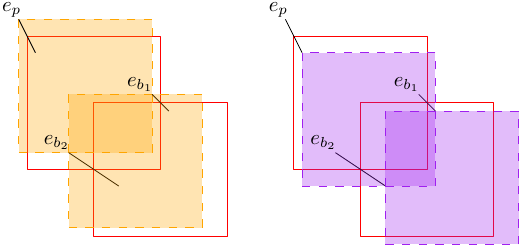}
    \caption{Sliding $\mathcal H_1$ and $\mathcal H_2$ along $e_p$, $e_{b_1}$ and $e_{b_2}$.}
    \label{fig:special1}
\end{figure}

\textbf{Summary.} In all special configurations there is a constant number of events between any two vertices of $\mathcal{T}_2$. Therefore, the number of vertices of $\mathcal{T}_2$ is an upper bound on the number of special configuration events up to a constant factor. By 

Lemma~\ref{lem:upper_envelope_f},
the number of vertices of $\mathcal{T}_2$ is at most $O(n 2^{\alpha(n)})$, so there are at most $O(n 2^{\alpha(n)})$ special configuration events in total.
\end{proof}

\begin{lemma} \label{lem:special_configuration_g}
  Given a trajectory
  with $n$ vertices, one can compute all special configuration events
  in $O(n2^{\alpha(n)} \log^2 n)$ time.
\end{lemma}
\begin{proof}
We use the same notation as in the proof of Lemma~\ref{lem:special_configuration_f}.
We compute the set $\mathcal{T}_2$. We take a pair of consecutive elements $s_i$ and $s_{i+1}$. We compute the segment $e_p$ that contains $\T[s_i, s_{i+1}]$. We use the reach data structure from Lemma~\ref{lem:reach_data_structure} to compute the reach of $p$ and therefore compute the segment $e_q$. We use the upper envelope data structure in Tool~\ref{tool:ue} to query $e_b$ (or both $e_{b_1}$ and $e_{b_2}$). We let $p: [0,1] \to e_p$ be the linear function defined in the proof of Lemma~\ref{lem:special_configuration_f}.

If we are in special configuration 1 or 2, we check if the translation
is parallel to $e_q$ or $e_b$ respectively, in which case we return
the first point of $\T[s_i, s_{i+1}]$. Otherwise, we compute the function $\mathcal H(t)$ of squares parametrised by $t$. The square $\mathcal H(t)$ has its top-right, or top-left corner at $p(t)$ for special configurations 1 and 2 respectively. Then we track the segment formed by the bottom-right corner of $\mathcal H(t)$ as we vary $t$. We return the value of $t$ where the bottom-right corner of $\mathcal H(t)$ lies on $e_q$.

If we are in special configuration 3, we compute the function $\mathcal H_1(t)$ of a square with its top-right corner on $p(t)$. Then we compute the intersections $u(p(t))$ and $l(p(t))$ of $\mathcal H_1(t)$ with $e_{b_1}$ and $e_{b_2}$ respectively. We let $\mathcal H_2(t)$ be the square with its top edge of $u(p(t))$ and its left edge of $l(p(t))$. We track the segment formed by the bottom-right corner of $\mathcal H_2(t)$ as we vary $t$. We return the value of $t$ where the bottom-right corner of $\mathcal H_2(t)$ lies on $e_q$. 

Now we analyse the running time of this algorithm. Building
Tool~\ref{tool:ue} takes $O(n\alpha(n)\log n)$ time. This is dominated
by the  $O(n \alpha(n) \log^2 n)$ time it takes to compute
$\mathcal{T}_2$ (Corollary~\ref{cor:compute_all_reaches} and
Lemma~\ref{lem:upper_envelope_f}).
Between each pair $(s_i, s_{i+1})$, we query the reach data structure and the upper envelope data structure, which takes $O(\log^2 n)$ and $O(\log n)$ time respectively. Constructing the functions $p(t)$, $\mathcal H_1(t)$, $u(p(t))$, $l(p(t))$ and $\mathcal H_2(t)$ are constant sized problems and only takes constant time. Therefore, the time to compute $\mathcal{T}_2$ is $O(n \alpha(n) \log^2 n)$ we spend $O(\log^2 n)$ query time for each element of $\mathcal{T}_2$. Since the size of $\mathcal{T}_2$ is $O(n 2^{\alpha(n)})$ by 

Lemma~\ref{lem:special_configuration_f}, 
the total running time of this algorithm is $O(n 2^{\alpha(n)} \log^2 n)$.
\end{proof}

\subsection{Summary}
\label{sub:summary_event_types}

We summarise the results of Sections~\ref{apx:reach_f}-\ref{apx:special_configuration_f} in the table below. Putting it all together, we obtain Theorem~\ref{thm:total_events}.

  \begin{tabular}{lll}
                                 & \#events             & computation
                                                          time \\
    \hline
    Vertex events                & $O(n)$            & $O(n)$ \\
    Reach events                 & $O(n)$            & $O(n \log^2 n)$ \\
    Bounding box events          & $O(n)$            & $O(n \log^2 n)$ \\
    Bridge events                & $O(n)$            & $O(n \log^2 n)$ \\
    Upper envelope events        & $O(n2^{\alpha(n)})$  & $O(n \alpha(n) \log^2 n)$ \\
    Special configuration events \hspace{5mm} & $O(n 2^{\alpha(n)})$ \hspace{5mm} & $O(n 2^{\alpha(n)} \log^2 n)$.\\
  \end{tabular}

\begin{theorem}
  \label{thm:total_events}
  The trajectory $\mathcal{T}_3$ has $O(n2^{\alpha(n)})$ vertices, and can be constructed in time $O(n2^{\alpha(n)}\log^2 n)$. 
\end{theorem}

\subsection{Computing the optimal subtrajectory}
\label{sub:Computing_a_longest_subtrajectory}

By Theorem~\ref{thm:starting_point_guaranteed_to_be_in_t3} there is a
longest $2$-coverable trajectory that starts at a point $p \in S$. By
Theorem~\ref{thm:total_events} this set $S$ has size
$m=O(n2^{\alpha(n)})$ and we can compute it in $O(n2^{\alpha(n)} \log^2 n)$
time. Using Corollary~\ref{cor:compute_all_reaches} we can compute a
longest 2-coverable subtrajectory starting at each point in $S$ in
$O(n\log n + m\log^2 n) = O(n2^{\alpha(n)}\log^2 n)$ time. We therefore
obtain the following result:

\begin{theorem}
  \label{thm:longest_2-coverable}
  Given a trajectory $\T$ with $n$ vertices, there is an
  $O(n2^{\alpha(n)}\log^2 n)$ time algorithm to compute a longest
  2-coverable subtrajectory of $\T$.
\end{theorem}

\section{Concluding Remarks}

We presented algorithms to decide if a set of segments is $k$-coverable for $k = 3, 4$, data structures for answering if subtrajectories are $k$-coverable for $k = 2, 3$, and algorithms to compute the longest $k$-coverable subtrajectory for $k=1,2$. One open problem is whether we can extend our algorithms to larger values of $k$. Another open problem is whether we can improve the bounds on the number of starting points of the longest 2-coverable subtrajectory, and whether we can compute them more efficiently. 

\bibliographystyle{plain}
\bibliography{bibliography.bib}

\end{document}